\documentclass{IEEEtran}
\usepackage{graphics}
\usepackage{stfloats}
\usepackage{epsfig}
\usepackage{times}
\usepackage{amsthm}
\usepackage{amssymb}
\usepackage{amsmath}
\usepackage{indentfirst}
\usepackage{epstopdf}
\usepackage{multirow}
\usepackage{float}
\usepackage{color}
\usepackage{enumerate}
\usepackage{subfigure}
\usepackage{comment}
\usepackage{cite}
\usepackage{caption}
\usepackage[ruled,vlined]{algorithm2e}
\usepackage{algpseudocode}
\usepackage{flushend}

\usepackage{tcolorbox}

\title{Optimal Synthesis of Opacity-Enforcing Supervisors for Qualitative and Quantitative Specifications}

\author{Yifan~Xie,~\IEEEmembership{Student Member,~IEEE,}
	Xiang~Yin,~\IEEEmembership{Member,~IEEE,}
    Shaoyuan~Li,~\IEEEmembership{Senior Member,~IEEE}
	\thanks{This work was  supported by the National Natural Science Foundation of China (62061136004, 61803259, 61833012) and by Shanghai Jiao Tong University Scientific and Technological Innovation Funds.
	}
	\thanks{Yifan Xie, Xiang Yin and Shaoyuan Li  are with Department of Automation
		and Key Laboratory of System Control and Information Processing,
		Shanghai Jiao Tong University, Shanghai 200240, China.
		{\tt\small  \{xyfan1234,yinxiang,syli\}@sjtu.edu.cn.} (Corresponding Author: Xiang Yin)}
}
\newtheorem{mydef}{Definition}
\newtheorem{mythm}{Theorem}
\newtheorem{myprob}{Problem}

\newtheorem{mycol}{Corollary}
\newtheorem{mypro}{Proposition}
\newtheorem{myexm}{Example}
\newtheorem{remark}{Remark}

\begin{document}
	
	\maketitle

	\begin{abstract}
	In this paper, we investigate both qualitative and quantitative synthesis of optimal privacy-enforcing supervisors for partially-observed discrete-event systems. 
	We consider a dynamic system whose information-flow is partially available to an intruder, which is modeled as a passive observer. We assume that the system has a ``secret" that does not want to be revealed to the intruder. 
	Our goal is to synthesize a supervisor that controls the system in a \emph{least-restrictive} manner such that the closed-loop system meets the privacy requirement.
	For the qualitative case, we adopt the notion of \emph{infinite-step opacity} as the privacy specification by requiring that the intruder can never determine for sure that the system is/was at a secret state for any specific instant. 
	If the qualitative synthesis problem is not solvable or the synthesized solution is too restrictive, 
	then we further investigate the quantitative synthesis problem so that the secret is revealed  (if unavoidable) \emph{as late as possible}.  
	Effective algorithms are provided to solve both the qualitative and quantitative synthesis problems. 
	Specifically, by building suitable information structures that involve information delays, we show that the optimal qualitative synthesis problem can be solved as a safety-game.  
	The optimal quantitative synthesis problem can also be solved as an optimal total-cost control problem over an augmented information structure.   
    Our work provides a complete solution to the standard infinite-step opacity control problem, which has not been solved without assumption on the relationship between controllable events and observable events. 
    Furthermore,  we generalize the opacity enforcement problem to the numerical setting by introducing the secret-revelation-time as a new quantitative measure. 
	\end{abstract}
	\begin{IEEEkeywords}
    Opacity, Supervisory Control, Discrete Event Systems, Optimal Control.
    \end{IEEEkeywords}
	
	\section{Introduction}
	\IEEEPARstart{P}{rivacy} and security issues have been becoming increasingly more important concerns in Cyber-Physical Systems (CPS) as communications and information exchanges among smart devices may cause information leakage  that threatens the system. Formal model-based methods provide rigorous, algorithmic and correct-by-construction  approaches towards the analysis and design of safety-critical CPS  whose  security and privacy demands are ever-increasing. 
	In this paper,  we investigate a formal information-flow security property called \emph{opacity} in the context of  Discrete-Event Systems (DES).
	Roughly specking,  a system is \emph{opaque} if its ``secret" can never be revealed to a malicious intruder that can access the information-flow of the system.
	The notion of opacity is essentially a confidentiality property that captures the plausible deniability for the system's secret \cite{bryans2008opacity}.

	In the context of DES, opacity has been studied very extensively in the past few years; see, e.g., some recent works \cite{berard2015probabilistic,keroglou2016probabilistic,tong2017verification,chen2017quantification,zhang2017infinite,basile2018algebraic,noori2018compositional,cong2018line,liu2020verification,saadaoui2020current,mohajerani2020transforming,lefebvre2020privacy,mohajerani2020compositional,lin2020information} and the survey papers \cite{jacob2016overview,lafortune2018history}.  
	Especially, to characterize different types of security requirements, different notions of opacity are proposed in the literature.
	For example, current-state opacity  (respectively,  initial-state opacity) requires that the intruder should never know for sure that the system is currently at (respectively, was initially from)  a secret state by utilizing the information available up to the current instant \cite{lin2011opacity}.
	In many situations, the intruder may further use future information  to better infer the security status of the system for some previous instants, which is essentially an information smoothing process. 
	To capture this scenario, the notions of $K$-step opacity \cite{saboori2012verification,falcone2015enforcement} and infinite-step opacity \cite{saboori2012verification,yin2017new,lan2020ifac} are proposed.
	Particularly, infinite-step opacity is the strongest one among all notions of opacity mentioned above,
	which requires that the intruder can never know for sure that the system is/was at a secret state for any specific instant even by using future information.

	When a system is verified to be non-opaque, one important problem  is to enforce opacity via some  mechanisms. 
	In general, there are two approaches for enforcing opacity: 
    one is to control the actual behavior of the system so that those secret-revealing behaviors can be avoided \cite{badouel2007concurrent,takai2008formula,ben2011supervisory,darondeau2014enforcing,partovi2020opacity,yin2016uniform,tong2018current,dubreil2010supervisory,zinck2020enforcing}
    the other one is to change the information-flow of the system so that the intruder can be ``cheated" or be ``confused" \cite{cassez2012synthesis,zhang2015max,rudie2019tase,yin2020synthesis,wu2018synthesis,wu2018privacy,ji2018enforcement,ji2019enforcing,liu2020k}.
	In particular, the first approach is essentially the supervisory control of opacity that aims to find a \emph{supervisor}  that restricts the behavior of the system dynamically such that the closed-loop system is opaque.
	For example, \cite{dubreil2010supervisory} studies the supervisory control problem for current-state opacity  assuming that all controllable events are observable and the intruder cannot observe more events than the supervisor.
	The work of  \cite{tong2018current} relaxes above assumptions but assumes that the intruder does not know the implementation of the supervisor. 
	In \cite{xie2020enforcing}, non-deterministic supervisors are used to enhance the plausible plausible deniability of the controlled system.
	Note that all the above mentioned works on opacity-enforcing supervisory control consider current-state opacity.

	In this paper, we study the problem of synthesizing optimal supervisors for \emph{infinite-step opacity} and its quantitative generalization. 
	Enforcing infinite-step opacity  is significantly more challenging than the standard current-state opacity enforcement problem.
	Specifically, in the infinite-step opacity setting, whether or not a secret can be revealed to the intruder not only depends on the information available currently, but also depends on the information in the future.
	To handle this future information, in the verification problem, we can look ahead in the original open-loop system by ``borrowing" the future information from the fixed plant model; see, e.g., \cite{saboori2012verification,yin2017new}.
	However, in the synthesis problem, the future information of the closed-loop system is \emph{unknown} and depends on the control policy in the future which is \emph{to be determined}.
	Hence, how to handle the dependency between the delayed information and the future control policy is the main difficulty in the synthesis of infinite-step-like opacity.
	In this paper, we propose  effective approaches that solve  both the qualitative and quantitative versions of the infinite-step opacity control synthesis problem. Our main contributions are twofold:  
	\begin{itemize}
	    \item  
	    First, we  provide  a new synthesis algorithm that solves the standard qualitative infinite-step opacity supervisory control problem without any assumption on the relationship between controllable events and observable events. Our approach is based on the generic structure of bipartite transition systems (BTS) \cite{yin2016uniform} and a new type of information-state that effectively captures the issue of information delay while avoiding the dependence on  future control policy. 
	    In particular, we show that the proposed synthesis algorithm is sound and complete, and the resulting supervisor is maximally-permissive in terms of language inclusion. 
	    Therefore, for the qualitative part, we completely solve the standard infinite-step opacity control problem, which was only partially solved in the literature under restrictive assumptions. 
	    \vspace{5pt}
	    \item 
	    Furthermore, we investigate a quantitative version of the opacity-enforcing control problem by introducing a secret revelation cost based on the notion of \emph{secret-revelation-time}, i.e., the earlier the secret is revealed, the higher the cost will be. This leads to a numerical generalization of infinite-step opacity (as well as the notion of $K$-step opacity \cite{saboori2012verification,falcone2015enforcement}). The control objective is to reveal each visit of secret states (if unavoidable) \emph{as late as possible}.
	    By suitably augmenting timing information into the information state space, we show that this problem can also be solved effective as an optimal worst-case total-cost control problem.   
	    Our approach provides a new angle for infinite-step opacity synthesis by using the secret-revelation-time as a quantitative measure. 
   \end{itemize}
 
	Our work is also related to several  works  in the  literature. 
	Regarding the qualitative synthesis problem,  as we mentioned earlier, most of the existing works on opacity-enforcing supervisory control only consider current-state opacity.
	One exception is \cite{saboori2011opacity}, where the authors propose a method to enforce infinite-step opacity by synthesizing \emph{a set of supervisors} that run synchronously.
	However, our approach synthesizes a single supervisor directly.
	More importantly,  the approach in \cite{saboori2011opacity} is based on the restrictive assumption that all controllable events are observable; this assumption is also relaxed in our approach. 
	In  our recent work \cite{yin2020synthesis}, we consider the synthesis of dynamic masks for infinite-step opacity. However, dynamic masks can only change the observation of the system, while supervisor can change the \emph{actually behavior} of the system without interfering the information-flow directly. Hence, different information-state updating rules are proposed here to handle the control problem. In terms of quantification of opacity,  most of the existing works  focus on qualifying \emph{how opaque} the system is in terms of probability measure; see, e.g.,  \cite{saboori2014current,berard2015probabilistic,chen2017quantification,yin2019infinite,deng2020fuzzy,liu2020notion}.  
	In the very recent work \cite{lefebvre2020exposure},  
	the authors proposed the concept opacity revelation time to characterize how long the initial secret is kept. 
	This concept is closely related to our secret-revelation-time. 
	However, here we essentially consider a delayed-state-estimation problem for each visit of secret states, which is much more involved than the initial-state-estimation problem considered in \cite{lefebvre2020exposure}. Furthermore, we consider the control synthesis problem, while \cite{lefebvre2020exposure} only considers the verification problem.

	The remaining part of the paper is organized as follows.
	Section~\ref{sec:2} presents dome necessary preliminaries.
	In Section~\ref{sec:3}, we formulate the qualitative opacity-enforcing control problem.
	In Section~\ref{sec:4}, we present a new class of information states for infinite-step opacity.
	In Section~\ref{sec:5}, we first define a bipartite transition systems and then present a synthesis algorithm that returns a maximally-permissive partial-observation supervisor to enforce infinite-step opacity.
	Furthermore, we quantify infinite-step opacity via secret-revelation-time and solve the quantitative synthesis problem in section~\ref{sec:6}.
	Finally, we conclude the paper in Section~\ref{sec:7}. 
	Some preliminary results for the qualitative part was presented in \cite{xie2020sup} without proof. 
	This paper presents complete proofs as well as detailed explanations.
	Moreover,  we further investigate the quantitative synthesis problem by quantifying the secret-revelation-time.  
	The techniques for solving the qualitative and quantitative problems are also different.

	\section{Preliminary}\label{sec:2}
    \subsection{System Model}
    We assume basic knowledge of DES and use common notations; see, e.g., \cite{Lbook}. 
    Let $\Sigma$ be an alphabet. 
    A string $s=\sigma_1\cdots\sigma_n$ is a finite sequence of events; $|s|=n$ denotes its length. 
    We denote by $\Sigma^*$ the set of all strings over $\Sigma$ including the empty string $\epsilon$ whose length is zero. 
    A language $L\subseteq \Sigma^*$ is a set of strings and we define $\overline{L}=\{u\in \Sigma^*: \exists v\in \Sigma^*, uv\in L   \}$ as the prefix-closure of $L$. 
    For the sake of simplicity, we write $s\leq t$ if $s\in \overline{\{t\}}$;  
    also we write $s<t$ if $s\leq t$ and $s\neq t$.
    
    A DES is modeled as a deterministic finite-state automaton
    \begin{equation}\label{DES}
        G=(X,\Sigma,\delta,x_0),\nonumber
    \end{equation}
    \noindent where $X$ is the finite set of states, $\Sigma$ is the finite set of events, $\delta:X\times\Sigma\rightarrow X$ is the partial transition function, where $\delta(x,\sigma)=y$ means that there is a transition labeled by event $\sigma$ from state $x$ to $y$, and $x_0\in X$ is the initial state.
    The transition function can also be extended to $\delta:X\times \Sigma^*\to X$ in the usual manner \cite{Lbook}.
    For simplicity, we write $\delta(x,s)$ as $\delta(s)$ when $x=x_0$.
    The language generated by $G$ is defined by $\mathcal{L}(G):=\!{\{s\in \Sigma^*: \delta(s)!}\}\!$, where $!$ means ``is defined''.

    When the system is partially observed, $\Sigma$ is partitioned into two disjoint sets: 
    \[
    \Sigma=\Sigma_o\dot{\cup}\Sigma_{uo},
    \]
    where $\Sigma_o$ is the set of observable events and $\Sigma_{uo}$ is the set of unobservable events.
    The natural projection $P:\Sigma^*\rightarrow \Sigma_o^*$ is defined recursively by
    \[
    P(\epsilon)=\epsilon  \text{ and }
    P(s\sigma)=\left\{
    \begin{aligned}
    &P(s)\sigma &\text{if }& \sigma\in\Sigma_o \\
    &P(s)       &\text{if }& \sigma\in\Sigma_{uo}
    \end{aligned}
    \right..
    \]
    For any observation $\alpha\beta\in P(\mathcal{L}(G))$, we define $\hat{X}_G(\alpha\mid\alpha\beta)$ as the \emph{delayed state estimate} that captures the set of all possible states the system could be in at the instant when $\alpha$ is observed given the entire observation $\alpha\beta$, i.e.,
    \begin{equation} 
    \hat{X}_{G}(\alpha\mid\alpha\beta):=
    \left\{
    \delta(s)\! \in\! X: 
     st\!\in\! \mathcal{L}(G) ,
    P(s)\!=\!\alpha, P(st)\!=\!\alpha\beta 
    \right\} .\nonumber
    \end{equation}
    For simplicity, we define $\hat{X}_G(\alpha):=\hat{X}_G(\alpha\mid\alpha)$ as the current state estimate upon the occurrence of $\alpha$.

    \subsection{Secret and Intruder}
    We assume that the privacy specification of the  system is captured by a set of secret states $X_S\subseteq X$.
    We consider an intruder modeled as a passive observer that may also observe the occurrences of observable events.
    Then  the intruder can infer whether or not the system was/is at a secret state based on the information-flow available. 
    Specifically, we use the notion of  \emph{infinite-step opacity} to describe the privacy requirement of the system, 
    which says that for any  string that leads to a secret state, the intruder can never determine for sure whether the system is/was at a secret state no matter what future information is generated.
    \begin{mydef}\label{def:inf-opa}
    Given system $G$, a set of observable events $\Sigma_o$ and a set of secret states $X_S$, system $G$ is said to be infinite-step opaque w.r.t.\ $X_S$ and $\Sigma_o$ if
    \begin{equation}
    \forall\alpha\beta\in P(\mathcal{L}(G)): \hat{X}_{G}(\alpha\mid\alpha\beta)\nsubseteq X_S.\nonumber
    \end{equation}
    \end{mydef}

    The following example illustrates the notion of infinite-step opacity.

    \begin{figure}
    \centering
    \subfigure[System $G$]{\label{SystemG}
    \includegraphics[width=0.152\textwidth]{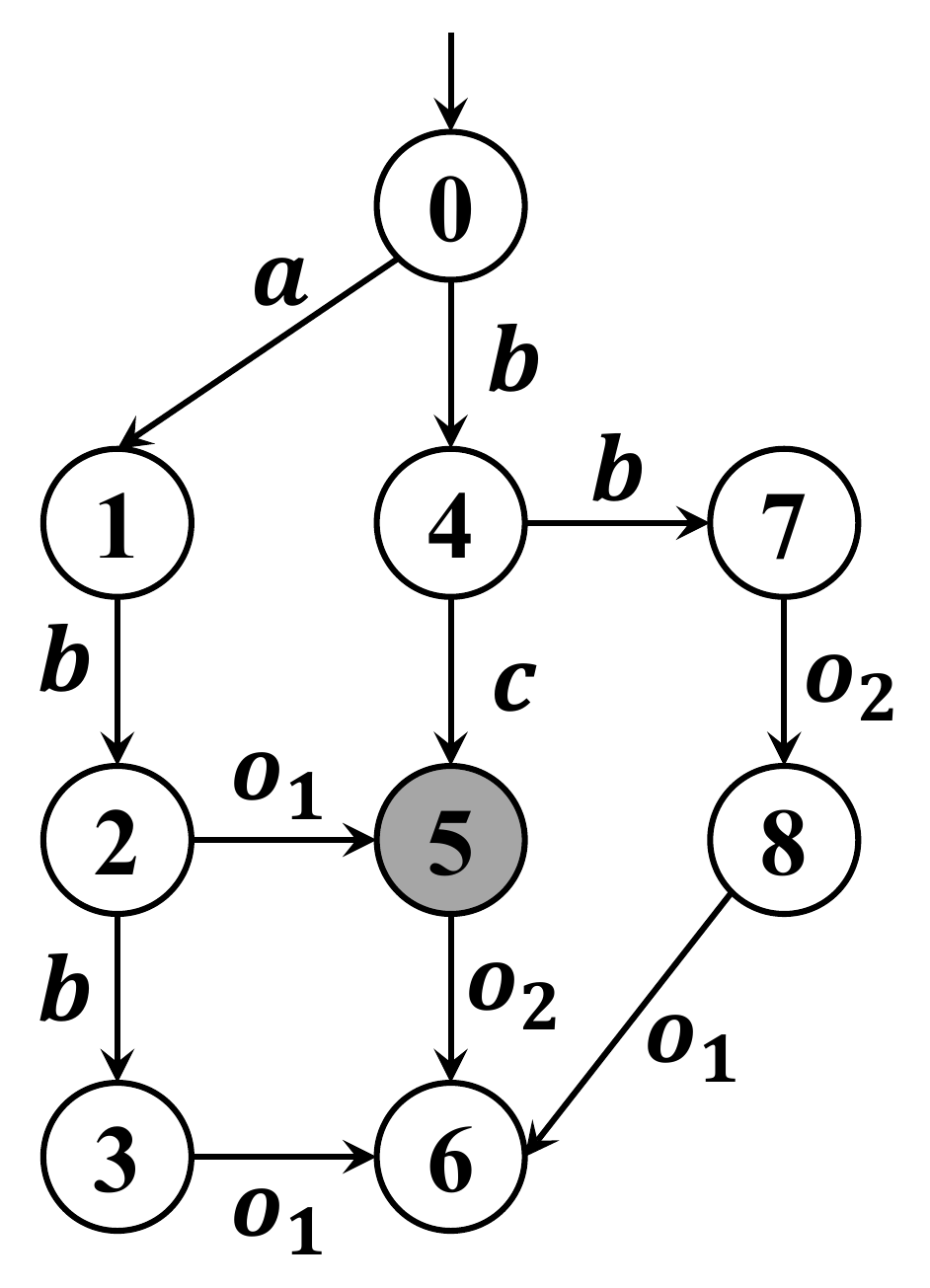}}
    \subfigure[System $G_1$]{\label{SystemG1}
    \includegraphics[width=0.152\textwidth]{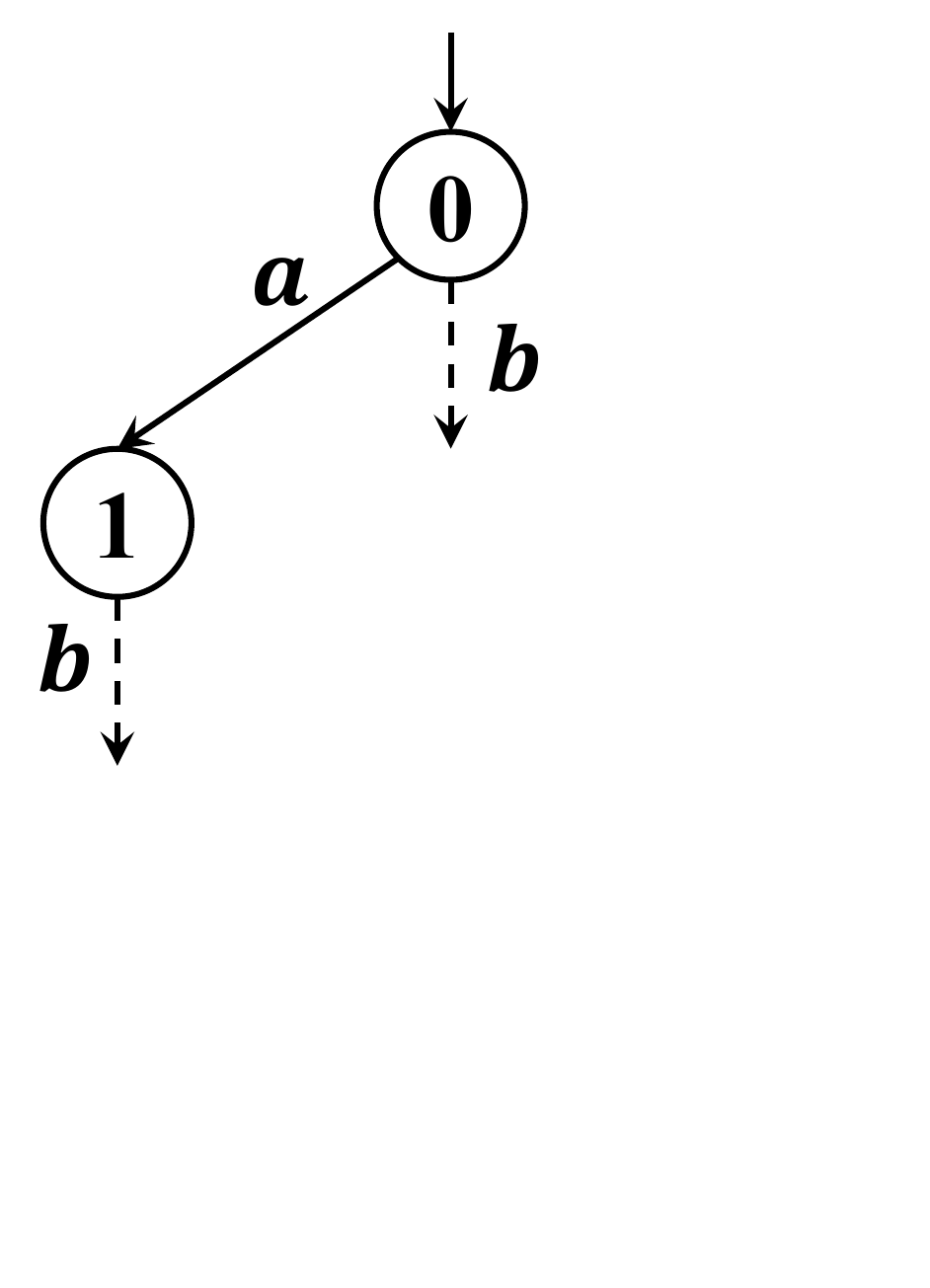}}
    \subfigure[System $G_2$]{\label{SystemG2}
    \includegraphics[width=0.152\textwidth]{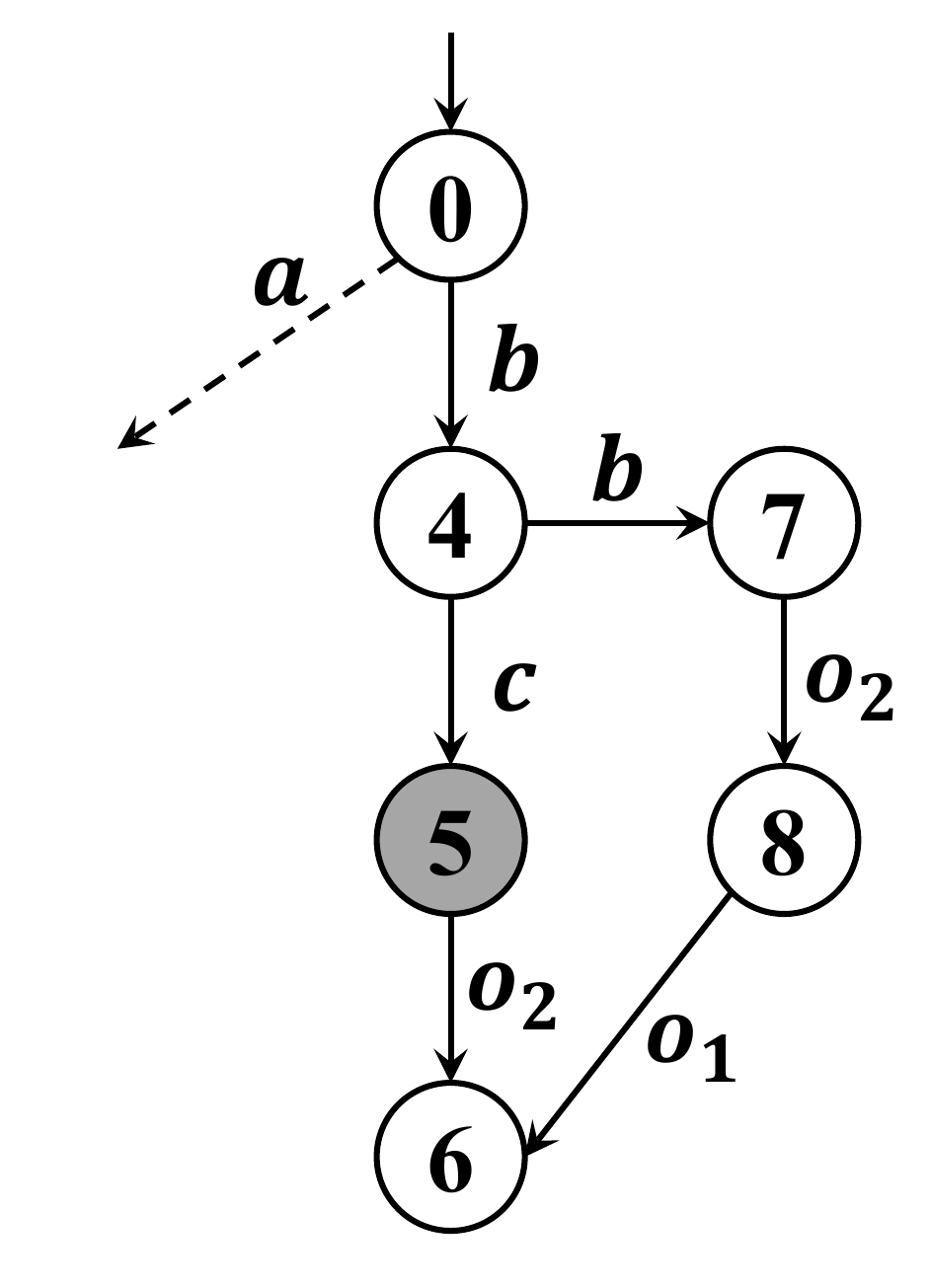}}
    \caption{For $G$: $\Sigma_o=\{o_1,o_2\},\Sigma_c=\{a,b,c\},X_s=\{5\}$.}
    \label{systemG}
    \end{figure}

    \begin{myexm}\label{exm:inf-opa}\upshape
    Let us consider system $G$ in Fig.~\ref{SystemG} with $\Sigma_o=\{o_1,o_2\}$ and $X_s=\{5\}$.
    Then $G$ is not infinite-step opaque w.r.t.\ $X_s$ and $\Sigma_o$.
    To see this, we consider string $abo_1$ with $P(abo_1)=o_1$ and string $abbo_1$ with $P(abbo_1)=o_1$.
    Then we know that $\hat{X}_G(o_1\mid o_1)=\{5,6\}$, i.e., when the intruder observe $o_1$ the system may be at state $5$ or state $6$.
    This says that when string $o_1$ is observed, the intruder cannot know for sure that the system is currently at the secret state. The system is actually current-state opaque.
    However, if we consider string $abo_1o_2$ with $P(abo_1o_2)=o_1o_2$, then the delayed state estimate is $\hat{X}_G(o_1\mid o_1o_2)=\{5\}\subseteq X_s$.
    This means that the intruder knows for sure that the system was at a secret state one step ago  when string $o_1o_2$ is observed.
    Upon the occurrence of $o_1o_2$, the secret state $5$ will be revealed to the intruder.
    Hence, the system is not infinite-step opaque.
    \end{myexm}

	\section{Qualitative  Privacy-Enforcing Control Problem}\label{sec:3}
    We start by the qualitative synthesis problem for infinite-step opacity. 
    When a given system $G$ is not opaque, i.e.,  the secret of the system can be revealed to the intruder, we need to enforce opacity on the system. One typical approach is to synthesize a  \emph{supervisor} that restricts the system's behavior to a sublanguage that satisfies the privacy requirement.

    In the framework of supervisory control, a supervisor can restrict the behavior of  $G$ by dynamically enabling/disabling some  events.
    In this setting, We assume that the events set is further  partitioned as 
    \[
    \Sigma=\Sigma_c\dot{\cup}\Sigma_{uc},
    \]
    where $\Sigma_c$ is the set of controllable events and $\Sigma_{uc}$ is the set of uncontrollable events.
    Controllable events are events that can be prevented from happening, or disabled by  the supervisor;  uncontrollable events cannot be disabled by the supervisor.
    A control decision $\gamma\in2^{\Sigma}$ is said to be admissable if $\Sigma_{uc}\subseteq\gamma$, namely uncontrollable events can never be disabled.
    We define $\Gamma=\{\gamma\in2^{\Sigma}:\Sigma_{uc}\subseteq\gamma\}$ as the set of admissable control decisions or control patterns.
    Since a supervisor can only make decisions based on its observation, a partial-observation supervisor is a function 
    \[
    S:P(\mathcal{L}(G))\rightarrow\Gamma.
    \]
    We use  notation $S/G$ to represent the closed-loop system under supervision and the language generated by $S/G$, denoted by $\mathcal{L}(S/G)$, is defined recursively by
    \begin{itemize}
        \item 
        $\epsilon\in \mathcal{L}(S/G)$; and 
        \item 
        for any $s\in \Sigma^*,\sigma\in \Sigma$, we have $s\sigma\in \mathcal{L}(S/G)$ iff
        $s\sigma\in \mathcal{L}(G),s\in \mathcal{L}(S/G)$ and $\sigma\in S(P(s))$.
    \end{itemize}

    In this paper, we want to synthesize a supervisor that disables events dynamically based on the observation trajectory such that the controlled system $S/G$ satisfies the privacy requirement, e.g., being infinite-step opaque for the qualitative setting.
    Note that, when the implementation of the supervisor becomes a public information, the intruder may further know that some behaviors in the original open-loop system are no longer feasible in the closed-loop system.
    Therefore, for any  observable string  $\alpha\beta\in P(\mathcal{L}(S/G))$, when the system is controlled by supervisor $S$, we define
    \begin{equation}
    \begin{split}
    \hat{X}_{S/G}(\alpha\mid\alpha\beta)\!:=\!
    \left\{
    \delta(s) \!\in\! X\!: 
    \begin{gathered}
      st \in \mathcal{L}(S/G), \\
    P(s)=\alpha, P(st)=\alpha\beta
    \end{gathered}
    \right\}
    \end{split}.\nonumber
    \end{equation}
    as the delayed state estimate in the closed-loop system.
    Similarly, we define $\hat{X}_{S/G}(\alpha)=\hat{X}_{S/G}(\alpha\mid\alpha)$.
    Then we say that the closed-loop system $S/G$ is  infinite-step opaque w.r.t.\ $X_S$ and $\Sigma_o$ if
    \begin{align}\label{eq:close-inf}
    \forall\alpha\beta\in P(\mathcal{L}(S/G)): \hat{X}_{S/G}(\alpha\mid\alpha\beta)\nsubseteq X_S.
    \end{align}
    Clearly, the less events enabled, the system is more likely to be opaque. 
    In the supervisory control framework, it is desirable that the supervisor enables as many events as possible to leave autonomy for the original system. Therefore, our goal is to find a maximally permissive supervisor (in terms of language inclusion).
    Then the qualitative privacy-enforcing supervisory control problem for infinite-step opacity  is formulated as follows.

    \begin{myprob}(Qualitative Opacity-Enforcing Control Problem)\label{problem}
    Given system $G$ and a set of secret states $X_S\subseteq X$, synthesize a partial-observation supervisor $S: P(\mathcal{L}(G))\rightarrow\Gamma$, such that:
    \begin{enumerate}[(i)]
	\item
	$S/G$ is infinite-step opaque w.r.t.\ $X_S$ and $\Sigma_o$; and
    \item
    For any other supervisor $S'$ satisfying (i), we have $\mathcal{L}(S/G)\not\subset\mathcal{L}(S'/G)$.
    \end{enumerate}
    \end{myprob}
    The second condition implies that the synthesized supervisor is maximal in the sense that its permissiveness cannot be improved anymore.
    As we will see in the following example, such maximal solution is not unique in general.

    \begin{myexm}\label{exm:111}\upshape
    To enforce infinite-step opacity for system $G$ shown in Fig.~\ref{SystemG}, we need to find a supervisor such that the close-loop behavior of the controlled system is opaque.
    Let $\Sigma_c=\{a,b,c\}$ be the set of controllable events.
    By disabling event $b$ initially, we can get an infinite-step opaque system shown in Fig.~\ref{SystemG1}.
    By disabling event $a$ initially, we can get another infinite-step opaque system, shown in Fig.~\ref{SystemG2}. 
    However, the union of $\mathcal{L}(G_1)$ and $\mathcal{L}(G_2)$ is not a feasible solution  since the supervisor needs to  disable $b$ initially in $\mathcal{L}(G_1)$ but needs to enable $b$ initially in $\mathcal{L}(G_2)$.
    One can easily check that  both $\mathcal{L}(G_1)$ and $\mathcal{L}(G_2)$ satisfy condition (2) in Problem \ref{problem}. However, they are two incomparable maximal solutions. 
    \end{myexm}

    Finally, we introduce some necessary  operators that will be used for further developments.
    Let $q\in2^X$ be a set of states, $\gamma\in\Gamma$ be a control decision and $\sigma\in\Sigma_o$ be an observable event.
    The \emph{unobservable reach} of  $q\subseteq X$ under control decision $\gamma\subseteq\Sigma$ is defined by
    \begin{equation}\label{unobservable reach} 
    U\!R_\gamma(q):=
    \left\{
    \delta(x,w)\in X:
    x\in q, w \in (\Sigma_{uo}\cap\gamma)^*
    \right\}, 
    \end{equation}
    which is the set of states that can be reached from states in $q$ via unobservable strings allowed in $\gamma$.
    The observable reach of $q\subseteq X$ upon the occurrence of $\sigma\in\Sigma_{o}$ is defined by
    \begin{equation}\label{observable reach}
    N\!X_\sigma(q):=  \{    \delta(x,\sigma)\in X: x\in q  \}.
    \end{equation}
    Similarly, let $\rho\in 2^{X\times X}$ be a set of   state pairs. 
    Intuitively, for each $(x,x')\in \rho$, $x$ represents   where the system was from at some instants and 
    $x'$ represents where the system is currently at. 
    Let $\gamma \in \Gamma$ be a control decision and $\sigma\in \Sigma_{o}$ be an observable event.
    We define:
    \begin{align}
     \widetilde{U\!R}_{\gamma}(\rho)
    \!=\!&
    \left\{\!
    (x,\delta(x',w)  )\!\in\! X\!\times\! X\!:\!
    (x,x')\!\in\! \rho, w\!\in\! (\Sigma_{uo}\!\cap\!\gamma)^*
    \!\right\} \\    
    \widetilde{N\!X}_{\sigma}(\rho)
    \!=\!&
    \left\{\!
    (x,\delta(x',\sigma) )\!\in\! X\!\times\!X \!:\!  (x,x')\!\in\! \rho   \right\}\\
    \odot_{\gamma}(q)
    \!=\!&
    \left\{\!
    (x,\delta(x,w) )\in X\!\times\!X \!:\!  x\in q, w\in(\Sigma_{uo}\!\cap\! \gamma)^*
    \!\right\}  
    \end{align} 
    Intuitively, $\widetilde{U\!R}_{\gamma}(\rho)$ and $\widetilde{N\!X}_{\sigma}(\rho)$, respectively, modify the unobservable reach and the observable reach, by not only tracking all possible current states, but also tracking where they come from. 
    Also, $\odot_{\gamma}(q)$ maps $q$ to a set of state pairs such that, in each pair of states, the first state can reach the second state via  enabled  but unobservable strings.

    \begin{myexm}\upshape
    Let us consider system $G$ in Fig.~\ref{SystemG} again.
    Let $q=\{1,5\}\in 2^X$ be a set of states and $\gamma=\{o_1,o_2,a,b\}\in \Gamma$ be a control decision.
    Then we have $U\!R_\gamma(\{1,5\})=\{1,2,3,5\}$.
    Upon the occurrence of enabled observable event $o_1\in \Sigma_o\cap\gamma$, we know that $N\!X_{o_1}(\{1,2,3,5\})=\{5,6\}$.
    Let $q=\{(0,1),(4,5)\}\in 2^{X\times X}$ be a set of state pairs, 
    which represents that the system is either (i) currently at state $1$ from state $0$; or (ii) currently at state $5$ from state $4$.
    Let $\gamma=\{o_1,o_2,a,b\}\in\Gamma$ be a control decision.
    Then we have $\widetilde{U\!R}_{\{o_1,o_2,a,b\}}(\{(0,1),(4,5)\})=\{(0,1),(0,2),(0,3),(4,5)\}$.
    Upon the occurrence of enabled observable event $o_1\in\Sigma_o\cap\gamma$, we have $\widetilde{N\!X}_{o_1}(\{(0,1),(0,2),(0,3),(4,5)\})=\{(0,5),(0,6)\}$.
    Besides, let $q=\{1,2,3,5\}\in 2^X$ and $\gamma=\{o_1,o_1,a,b\}$, we have $\odot_{\gamma}(q)=\{(1,1),(1,2),(1,3),(2,2),(2,3),(3,3),(5,5)\}$.
    \end{myexm}
	
	\section{Information-state and its flow}\label{sec:4}
	To enforce infinite-step opacity, the main difficulty is that the state estimation is \emph{delayed} and one can use future observation to improve its knowledge about the system at some previous instants.  
	In this section, we study how information evolves in the closed-loop system when delayed information is involved.
	%.
	\subsection{Notion of Information-State}
In the synthesis of partially-observed systems, $2^X$ is usually chosen as the \emph{information-states} representing current information of the system.
This information-state has been shown to be suitable for current-state opacity enforcement.
However, infinite-step opacity requires that the secret cannot be revealed at any instant currently and in the future; hence, $2^X$ is not sufficient enough.

    Note that, the requirement of infinite-step opacity in Equation~\eqref{eq:close-inf} can be equivalently written as
    \begin{align}\label{eq:reform}
    \forall\alpha \in P(\mathcal{L}(S/G)),\forall \alpha' \leq  \alpha : \hat{X}_{S/G}(\alpha' \mid \alpha)\nsubseteq X_S.
    \end{align}
    This reformulation suggests that, to enforce infinite-step opacity, it is sufficient to capture the delayed state estimates of all previous instants.
    Therefore, instead of using $2^X$,  we use another set of information-states  defined by
    \[
    I:=2^X\times 2^{2^{X\times X}}.
    \]
    Then each information-state $\imath\in I$ is in the form of $\imath=(C(\imath),D(\imath))$, where
    \begin{itemize}
	\item
	The first component $C(\imath)\in 2^X$ is a set of states that captures the current state estimate of the system; and
	\item
	The second component $D(\imath)\in 2^{2^{X\times X}}$ is a set of state pairs that captures all possible delayed state estimates in history.
    Specifically,   each element $\rho\in  D(\imath)$ is in the form of $\{(x_1,x_1'),\dots, (x_k,x_k')\}$, 
    where $x_i\in X$ represents the state of the system at some previous instants and $x'_i\in X$ represents the  current state of the system, 
    and $\rho$ contains all possibilities for that instant. 
    Then  $D(\imath)=\{\rho_1,\dots,\rho_n\}$ essentially contains all such sets of pairs for all previous instants.
    \end{itemize}

    \subsection{Information-State Updating Rule}
    Suppose that the system's  information-state is $\imath=(C(\imath),D(\imath))$.
    Then, upon the occurrence of an observable event  $\sigma\in \Sigma_{o}$ (should also be allowed by the previous control decision), the supervisor will issue a new  control decision $\gamma\in\Gamma$.
    Therefore, $(\sigma,\gamma)$ is our new information about the system and the information-state  $\imath=(C(\imath),D(\imath))$ needs to be updated to $\imath'=(C(\imath'),D(\imath'))$ as follows:
    
\begin{center}
        \tcbset{width=8.5cm,lefttitle=1.6cm,title=Information-State Updating Rule,colback=white}
    \begin{tcolorbox} 
    \vspace{-8pt}
    \begin{equation}
    	\!\! \!\!\!\!	\left\{
    \begin{aligned}\label{Updating}
    C(\imath')
    = &U\!R_{\gamma}(N\!X_{\sigma}(C(\imath)))\\
    D(\imath')
     = &\{ \widetilde{U\!R}_{\gamma}(\widetilde{N\!X}_{\sigma}(\rho))\in 2^{X\times X}:\rho\in D(\imath)\} \\
    &\cup \{ \odot_{\gamma}(C(\imath'))\}\\
    \end{aligned}
    \right.
    \end{equation}
    \end{tcolorbox}
\end{center}

    Intuitively, the first equation simply updates the current state estimate of the system, while the second equation updates all possible delayed state estimates (pairs) of previous instants and adds the current state estimate to the history.
    Then we consider a controlled system $S/G$.
    Let $\alpha=\sigma_1\sigma_2\cdots\sigma_n\in P(\mathcal{L}(S/G))$ be an observable string.
    Then the following information states evolution will happen
    \begin{equation}\label{In-Evo1}
    \imath_0\xrightarrow[\quad]{(\sigma_1,S(\sigma_1))}\imath_1\xrightarrow[\quad]{(\sigma_2,S(\sigma_1\sigma_2))}\cdots\xrightarrow[\quad]{(\sigma_n,S(\sigma_1\cdots\sigma_n))}\imath_n
    \end{equation}
    where $\imath_0=(U\!R_{S(\epsilon)}(\{x_0\} ),\{\odot_{S(\epsilon)}(U\!R_{S(\epsilon)}( \{x_0\} ))\})$ represents the initial information-state
    and each
    $\imath_{i-1}\xrightarrow[\quad]{(\sigma_{i},S(\sigma_1\cdots\sigma_i))}\imath_i$ means that
    $\imath_i$ is obtained from $\imath_{i-1}$ with new information $(\sigma_{i},S(\sigma_1\cdots\sigma_i))$ according to the updating rule in Equation~\eqref{Updating}.
    We denote by $\mathcal{I}(\alpha)$ the information state reached by $\alpha$, i.e., $\mathcal{I}(\alpha)=\imath_n$.

    We illustrate the above information updating procedure by the following example.
	
	\begin{myexm}\label{exm4}\upshape
		Let $S$ be the supervisor that results in   closed-loop system $G_2$ in Fig.~\ref{SystemG2}.
		Specifically, supervisor $S$ disables $a$ initially and then enables all events, i.e., $S(\epsilon)=\{o_1,o_2,b,c\}$ and $S(\alpha)=\Sigma$ for $\alpha\neq \epsilon$.
		Let us consider observable string $o_2o_1\in P(\mathcal{L}(S/G))$.
		Then we have
		$$
		\mathcal{I}(\epsilon)=(U\!R_{S(\epsilon)}(\{x_0\}),\{\odot_{S(\epsilon)}(U\!R_{S(\epsilon)}(\{x_0\}))\}),
		$$
		where
		$$
		\begin{aligned}
		C(\mathcal{I}(\epsilon))=&\{0,4,5,7\}\\
		D(\mathcal{I}(\epsilon))=&
		\left\{\!\!
		\left\{\!\!\!
		\begin{array}{c c}
		(0,0),\!(0,4),\!(0,5),\!(0,7),\\
		(4,4),\!(4,5),\!(4,7),\!(5,5),\!(7,7)
		\end{array}\!\!\!
		\right\}\!\!
		\right\}.
		\end{aligned}
		$$
		Once event $o_2$ is observed,  new control decision $S(o_2)=\{o_1,o_2,a,b,c\}$ is made, and the updated information-state is $\mathcal{I}(o_2)$, where
		$$
		\begin{aligned}
		C(\mathcal{I}(o_2))=&U\!R_{S(o_2)}(N\!X_{o_2}(C(\mathcal{I}(\epsilon))))\\
		=&\{6,8\}\\
		D(\mathcal{I}(o_2))=&\{\widetilde{U\!R}_{S(o_2)}(\widetilde{N\!X}_{o_2}(\rho)):\rho\in D(\mathcal{I}(\epsilon))\}\\
		&\cup\{\odot_{S(o_2)}(C(\mathcal{I}(o_2))\}\\
		=&
		\left\{\!\!\!\!
		\begin{array}{c c}
		\{(0,6),\!(0,8),\!(4,6),\!(4,8),\!(5,6),\!(7,8)\},\\
		\{(6,6),(8,8)\}
		\end{array}
		\!\!\!\!
		\right\}.
		\end{aligned}
		$$
		Then event $o_1$ is observed and new control decision $S(o_2o_1)=\{o_1,o_2,a,b,c\}$ is made.
		The  information-state is then updated to $\mathcal{I}(o_2o_1)$, where
		$$
		\begin{aligned}
		C(\mathcal{I}(o_2o_1))=&U\!R_{S(o_2o_1)}(N\!X_{o_1}(C(\mathcal{I}(o_2))))\\
		=&\{6\}\\
		D(\mathcal{I}(o_2o_1))
		=&\{\widetilde{U\!R}_{S(o_2o_1)}(\widetilde{N\!X}_{o_1}(\rho)):\rho\in D(\mathcal{I}(o_2))\}\\
		&\cup\{\odot_{S(o_2o_1)}(C(\mathcal{I}(o_2o_1))\}\\
		=&\{\!\{(0,6),\!(4,6),\!(7,6)\},\{(8,6)\},\{(6,6)\}\!\}.
		\end{aligned}
		$$
	    \end{myexm}
	
	\subsection{Property of the Information-State}
	In the following result, we formally show that the proposed information-state updating rule indeed yields the desired delayed state estimate in the controlled system.

	\begin{mypro}\label{prop:main}
		Let $S$ be a supervisor, $\alpha\in P(\mathcal{L}(S/G))$ be an observable string and $\mathcal{I}(\alpha)$ be the information-state reached.
		Then we have
		\begin{enumerate}[(i)]
			\item
			$C(\mathcal{I}(\alpha))=\hat{X}_{S/G}(\alpha)$; and
			\item
			$D(\mathcal{I}(\alpha))=\{\rho_{\alpha',\alpha}\in 2^{X\times X}:\alpha'\leq\alpha \}$, where
		\end{enumerate}
		\begin{equation}
		\rho_{\alpha',\alpha}  \!=\!
		\left\{
		(\delta(s),\delta(st))\!\in\! X \!\times\! X:
		\begin{gathered}
		  st\in \mathcal{L}(S/G), \\
		P(s)=\alpha',P(st)=\alpha
		\end{gathered}
		\right\}.
		\end{equation}
	\end{mypro}
    \begin{proof}
    The first component $C(\mathcal{I}(\alpha))$ is obtained by iteratively applying $U\!R$ and $N\!X$ operators,
    which gives (i)  according to \cite{yin2016uniform}.
    Next, we prove (ii) by induction on the length of $\alpha$.

    \noindent\emph{Induction Basis: }Suppose that $|\alpha\mid=0$, i.e., $\alpha=\epsilon$.
    Then we know that $\mathcal{I}(\epsilon)=\imath_0$.
    Let $\gamma=S(\epsilon)$ and we have 
    \begin{align}\label{D()}
    &D(\mathcal{I}(\epsilon))         \nonumber\\
    \!=\!&
    \{\odot_{\gamma}(U\!R_{\gamma}(\{x_0\}))\}       \nonumber\\
    \!=\!&
    \left\{  \left\{(x, \delta(x,w))\!\in\! X  \!\times\! X:
    \begin{gathered}
    x\!\in\! U\!R_{\gamma}(\{x_0\}),   w\! \in\!(\Sigma_{uo}\cap\gamma)^{*}
    \end{gathered}
    \right\}  \right\}            \nonumber\\
    \!=\!&
    \left\{\left\{(\delta(s),\delta(st))  \in X\times X:
    \begin{gathered}
    st\in(\Sigma_{uo}\cap\gamma)^{*}
    \end{gathered}
    \right\}\right\}                   \nonumber\\
    \!=\!&
    \left\{\left\{(\delta(s),\delta(st)) \in X\times X:
      st\!\in\!\mathcal{L}(S/G),   
    P(s)\!=\!P(st)\!=\!\epsilon 
    \right\}\right\} \nonumber\\
    \!=\!&
    \{\rho_{\epsilon,\epsilon}\}
    \end{align}
    That is, the induction basis holds.

    \noindent\emph{Induction Step: }
    Now, let us assume that (ii)  holds for $|\alpha\mid=k$.
    Then we want to prove the case of $\alpha\sigma$, where  $\sigma\in\Sigma_o\cap\gamma$.
    In the following equations, $\gamma=S(\alpha)$.
    By the updating rule in Equation~\eqref{Updating} and (i), we know that
    \begin{align}\label{DD}
    &D(\mathcal{I}(\alpha\sigma)) \\
    =&\{ \widetilde{U\!R}_{\gamma}(\widetilde{N\!X}_{\sigma}(\rho))\in 2^{X\times X}:\rho\in D(\mathcal{I}(\alpha))\} \nonumber\\
    &\cup \{ \odot_{\gamma}(\hat{X}_{S/G}(\alpha\sigma))\}  \nonumber
    \end{align}
    Since $|\alpha\mid=k$, by the induction hypothesis, we know
    \[
    D(\mathcal{I}(\alpha))=\{\rho_{\alpha',\alpha}\in 2^{X\times X}:\alpha'\leq\alpha \}
    \]
    Recall that
    \begin{align}
    \rho_{\alpha',\alpha}
    \!=\!&
    \left\{(\delta(s),\delta(st))\in X\times X:
    \begin{gathered}
      st\in \mathcal{L}(S/G), \\
    P(s)=\alpha', P(st)=\alpha
    \end{gathered}
    \right\}\nonumber
    \end{align}
    Therefore, we have
    \begin{equation}\label{D(1)}
    \begin{split}
    &\widetilde{U\!R}_{\gamma}(\widetilde{N\!X}_{\sigma}(\rho_{\alpha',\alpha}))\\
    =&
    \left\{(x_1,\delta(x_2,\sigma w)  ) \!\in\! X\!\times\! X \!:\!
    (x_1,x_2)\!\in\!\rho_{\alpha',\alpha},  w\!\in\!(\Sigma_{uo}\!\cap\!\gamma)^{*}
    \right\}\nonumber\\
    =&
    \left\{
    \begin{gathered}
    (\delta(s),\delta(st\sigma w)) \\
    \in   X\times X
    \end{gathered} :
    \begin{gathered}
     st\in\mathcal{L}(S/G),w\in(\Sigma_{uo}\cap\gamma)^{*},\\  
    P(s)=\alpha',  P(st)=\alpha
    \end{gathered}
    \right\}\nonumber\\
    =&
    \left\{(\delta(s),\delta(st'))\in X\times X:
    \begin{gathered}
      st'\in\mathcal{L}(S/G),    \\
    P(s)=\alpha',  P(st')=\alpha\sigma
    \end{gathered}
    \right\}\nonumber\\
    =&\rho_{\alpha',\alpha\sigma}
    \end{split}
    \end{equation}
    This further gives
    \begin{equation}\label{U}
    \begin{split}
    &\{ \widetilde{U\!R}_{\gamma}(\widetilde{N\!X}_{\sigma}(\rho))\in 2^{X\times X}:\rho\in D(\mathcal{I}(\alpha))\} \\
    =&\{ \widetilde{U\!R}_{\gamma}(\widetilde{N\!X}_{\sigma}(\rho_{\alpha',\alpha}))\in 2^{X\times X}:\alpha'<\alpha\sigma \}\\
    =&\{\rho_{\alpha',\alpha\sigma}\in 2^{X\times X}:\alpha'<\alpha\sigma \}
    \end{split}
    \end{equation}
    Moreover, we have
    \begin{align}\label{o}
    &\odot_{\gamma}(\hat{X}_{S/G}(\alpha\sigma))\\
    =&
    \left\{(x,\delta(x,w))\!\in\! X\!\times\! X :
    x\!\in\!\hat{X}_{S/G}(\alpha\sigma), w\!\in\!(\Sigma_{uo}\cap\gamma)^{*}
    \right\}\nonumber\\
    =&
    \left\{\!(\delta(s),\delta(sw))\!\in\! X\!\times\! X\!:\!
    \begin{gathered}
      s\!\in\!\mathcal{L}(S/G),\\
     w\!\in\!(\Sigma_{uo}\cap\gamma)^{*}, P(s)=\alpha\sigma
    \end{gathered}
    \right\}\nonumber\\
    =&
    \left\{\!(\delta(s),\delta(sw))\!\in\! X\!\times\! X\!:\!
    \begin{gathered}
      sw\in\mathcal{L}(S/G), \\
      P(s)=P(sw)=\alpha\sigma
    \end{gathered}
    \right\}\nonumber\\
    =&\rho_{\alpha\sigma,\alpha\sigma} \nonumber
    \end{align}
    Therefore, by combing Equation~\eqref{D(1)} and \eqref{o} we have
    \begin{align}\label{Y}
    D(\mathcal{I}(\alpha\sigma))=&\{\rho_{\alpha',\alpha\sigma}\in 2^{X\times X}:\alpha'<\alpha\sigma \}\cup\{\rho_{\alpha\sigma,\alpha\sigma}\}  \nonumber\\
    =&\{\rho_{\alpha',\alpha\sigma}\in 2^{X\times X}:\alpha'\leq \alpha\sigma\}
    \end{align}
    This completes the induction step, i.e., (ii) holds.
    \end{proof}

	Recall that each information-state is in the form of $\imath=(C(\imath),D(\imath))=2^X\times2^{2^{X\times X}}$,
	where $D(\imath)$ is a set of possible state-pairs.
	We define
	\begin{equation}\label{D1}
    D_1(\imath):=\{\{x\in X:(x,x')\in\rho\}:\rho\in D(\imath)\}
    \end{equation}
	as the set of its first components. 
	For example, for  $\mathcal{I}(o_2)=(\!\{6,\!8\},\!\{\{(0,6),\!(0,8),\!(4,6),\!(4,8),\!(5,6),\!(7,8)\},\!\{(6,6),
	\!(8,8)\}\!\}\!)$ in Example~\ref{exm4}, we have $D_1(\mathcal{I}(o_2))=\{\{0,4,5,7\},\{6,8\}\}$.
	The following result shows that $D_1(\imath)$ indeed captures all possible delayed state estimates.
	
	\begin{mycol}\label{cor:delay}
		Let $S$ be a supervisor, $\alpha\in P(\mathcal{L}(S/G))$ be an observable string and $\mathcal{I}(\alpha)$ be the information-state reached.
		Then we have
		\[
		D_1(\mathcal{I}(\alpha))=\{\hat{X}_{S/G}(\alpha' \mid \alpha)\in 2^X:\alpha'\leq\alpha \}.
		\]
	\end{mycol}
	\begin{proof}
	By Proposition~\ref{prop:main}, we know that $D(\mathcal{I}(\alpha))=\{\rho_{\alpha',\alpha}\in 2^{X\times X}:\alpha'\leq\alpha \}$.
   Therefore,
    \begin{equation}
    \begin{split}
    &D_1(\mathcal{I}(\alpha))\\
    =&
    \{\{x\in X:(x,x')\in\rho_{\alpha',\alpha}\}:\alpha'\leq\alpha \}\\
    =&
    \left\{\left\{\delta(s)\in X:
    \begin{gathered}
      st\in\mathcal{L}(S/G), \\
    P(s)=\alpha',P(st)=\alpha\}
    \end{gathered}\right\}
    :\alpha'\leq\alpha 
    \right\}\nonumber\\
    =&\{\hat{X}_{S/G}(\alpha'\mid\alpha)\in 2^X:\alpha'\leq\alpha \}
    \end{split}
    \end{equation}
    \end{proof}
	
	\section{Solving the Qualitative Synthesis Problem}\label{sec:5}
	
	\subsection{Bipartite Transition System}
	In Section~\ref{sec:4}, we have proposed a new type of information-state that can capture all possible delayed state estimates.
	Note that, the information-state updating rule as defined in Equation~\eqref{Updating} essentially consists of  two steps:
	the immediate observable reach when a new observable event occurs and the unobservable reach when a new control decision is issued.
	However, this is based on the assumption that supervisor $S$ is given.
	In the synthesis problem, the control decision at each instant is unknown and to be determined.
	Therefore, we need to separate these two updating steps clearly.
	To this end,
	we adopt  the generic structure of the \emph{bipartite transition system} (BTS) that was originally proposed in our previous work \cite{yin2016uniform} by incorporating the new information-state.
	
	\begin{mydef}
		A  bipartite transition system (BTS) $T$ w.r.t.\ $G$ is a 7-tuple.
		\begin{equation}\label{T}
		T=(Q_Y^T,Q_Z^T,h_{Y\!Z}^T,h_{ZY}^T,\Sigma_o,\Gamma,y_0^T),
		\end{equation}
		where
		\begin{itemize}
			\item
			$Q_Y^T\subseteq I$ is the set of $Y$-states.
			Therefore, a $Y$-state $y\in Q_Y^T$ is in the form of  $y=(C(y),D(y))$;
			\item
			$Q_Z^T\subseteq I\times\Gamma$ is the set of $Z$-states.
			For each $z\in Q_Z^T$,  $I(z)$ and $\Gamma(z)$ denote, respectively, the information state and the control decision, so that $z=(I(z),\Gamma(z))\!\in\!2^X\times 2^{2^{X\times X}}\times\Gamma$.
			For simplicity, we further write
			$z=(C(I(z)),D(I(z)), \Gamma(z))$ as $z=(C(z),D(z),\Gamma(z))$;
			\item
			$h_{Y\!Z}^T:Q_Y^T\times\Gamma\rightarrow Q_Z^T$ is the partial transition function from $Y$-states to $Z$-states satisfying the following constraint:  for any $h_{Y\!Z}^T(y,\gamma)=z$, we have
			\begin{equation}\label{hYZ}
			\left\{
			\begin{aligned}
			C(z)=& U\!R_{\gamma}(C(y)) \\
			D(z)=& \{ \widetilde{U\!R}_{\gamma}(\rho)\!\in\! 2^{X\times X}:\rho\!\in\! D(y)\} \cup \{ \odot_{\gamma}(C(z))\}  \\
			\Gamma(z)=& \gamma
			\end{aligned}
			\right.
			\end{equation}
			\item
			$h_{ZY}^T:Q_Z^T\times\Sigma\rightarrow Q_Y^T$ is the partial transition function from $Z$-states to $Y$-states satisfying the following constraint:  for any  $h_{ZY}^T(z,\sigma)=y$, we have $\sigma\in\Gamma(z)\cap\Sigma_o$ and
			\begin{equation}\label{hZY}
			\left\{
			\begin{aligned}
			C(y)=&N\!X_\sigma(C(z)) \\
			D(y)=&\{\widetilde{N\!X}_\sigma(\rho)\in 2^{X\times X}:\rho\in D(z)\}
			\end{aligned}
			\right.
			\end{equation}
			\item
			$\Sigma_o$ is the set of observable events of $G$;
			\item
			$\Gamma$ is the set of admissible control decisions of $G$;
			\item
			$y_0^T:=(\{x_0\},\{ \emptyset   \})\in Q_Y^T$ is the initial $Y$-state.
		\end{itemize}
	\end{mydef}
	
	The BTS  is essentially a game structure between the controller and the environment.
	When the controller picks a control decision $\gamma\in 2^\Sigma$ at a $Y$-state, the game moves to a $Z$-state.
	A transition from $Y$-state to $Z$-state is an unobservable reach under the issued control decision and remembers the control decision.
	When the environment picks an observation $\sigma\in \Sigma_o\cap\gamma$ at a $Z$-state, the game moves to a $Y$-state, and so forth. A transition from $Z$-state to $Y$-state is the observable reach.
	Transitions from $Z$-states to $Y$-states and transitions from $Y$-states to $Z$-states are indeed the information-state updating rule in Equation~\eqref{Updating}, but we separate the updating procedure into two parts.
	Specifically, for any $z\in Q_Z^T, y\in Q_Y^T,\gamma\in \Gamma,\sigma\in \Sigma_o$,
	we have $h_{Y\!Z}^T(h_{ZY}^T(z,\sigma),\gamma)= (\imath',\gamma)$, where
	$I(z)\xrightarrow[\quad]{(\sigma,\gamma)}\imath'$. 
	For simplicity, we write a transition as $h_{Y\!Z}$ whenever it is defined for some $h_{Y\!Z}^T$; the same for $h_{ZY}$. 
	The basic generic structure of the BTS was originally proposed in \cite{yin2016uniform}.
	Here we generalize the original BTS by using new information-states capturing delays rather than the current-state-type information-states   used in \cite{yin2016uniform}.
	
	\begin{figure*}[htbp]
		\centering
		\includegraphics[width=1\linewidth]{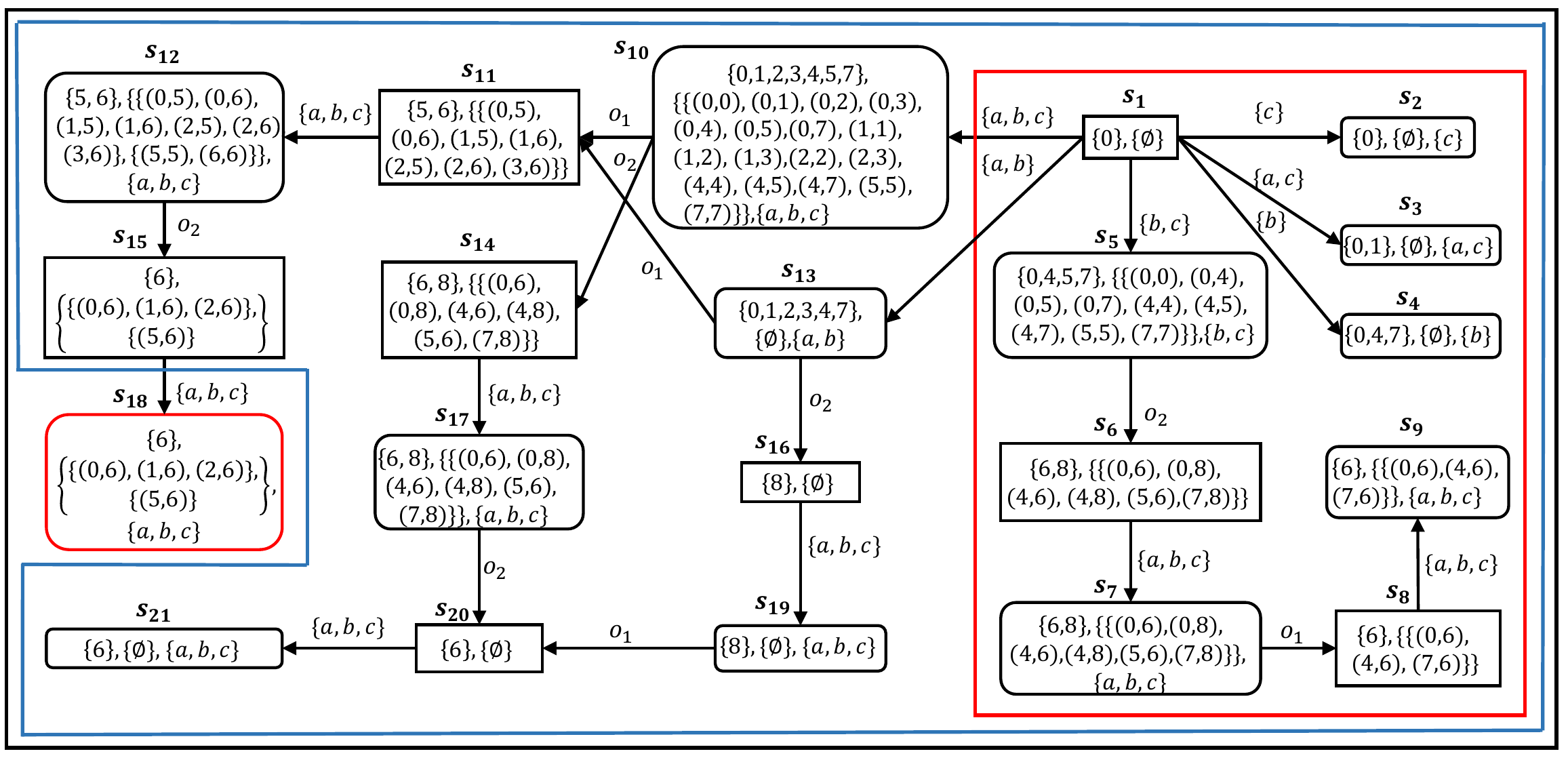}
		\caption{Example of the construction of the BTS.
			Rectangular states correspond to $Y$-states and rectangular with rounded corners states correspond to $Z$-states.
			$T_{total}$ is the entire system in the box marked with black lines, which is the largest BTS that enumerates all transitions, $T_0$ is the system in the box marked with blue lines and 
			$T^*$ is the system in the box marked with red lines.}
		\label{BTS}
	\end{figure*}

	\subsection{Supervisor Synthesis Algorithm} 
	
	By Equation~\eqref{eq:reform}  and Corollary~\ref{cor:delay},
	to make sure that the closed-loop system $S/G$ is infinite-step opaque,
	it suffices to guarantee that, for  any information-state $\imath$ reached, we have
	$\forall q\in D_1(\imath), q\nsubseteq X_S$.
	Therefore, we define
	\[
	Q_{rev}=\{z\in Q_Z:   q\in D_1( I(z)   ), q\subseteq X_S\}
	\]
	as the set of \emph{secret-revealing} $Z$-states.
	To synthesize a supervisor that enforces infinite-step opacity, the controlled system $S/G$ needs to guarantee that all  reachable information-states are not secret-revealing.
	Furthermore, as the supervisor can only play at $Y$-states, we need to make sure that
	(i) there is at least one choice at each $Y$-state; and (ii) all choices at each $Z$-state should be considered.
	Therefore, for each BTS $T$, we say that a state is \emph{consistent} if
	\begin{itemize}
		\item
		at least one transition is defined when it  is a $Y$-state;
		\item
		all feasible observations are defined when it is a $Z$-state.
	\end{itemize}
	Let $\mathcal{T}$ be the set of all  BTSs.
	For any  BTS $T\in \mathcal{T}$, we define $Q^T_{consist}$ as the set of consistent states in $T$.
	Also,  for a  set of states $Q\subseteq Q_Y^T\cup Q_Z^T$,
	we define $T|_{Q}$ as the restriction of $T$ to $Q$,
	i.e., $T|_{Q}$ is the BTS obtained by removing states not in $Q$ and their associated transitions  from $T$.
	
	In order to synthesize an opacity-enforcing supervisor,
	first, we construct the largest BTS that enumerates all possible transitions for each state,
	i.e., a transition is defined whenever it satisfies the constraints in $h_{Y\!Z}$ or $h_{ZY}$.
	We define $T_{total}\in\mathcal{T}$ as the largest BTS that includes all possible transitions.
	Second, we remove all secret-revealing states from $T_{total}$ and define
	\[
	T_0=T_{total}|_{Q\setminus Q_{rev}}.
	\]
	Then, we need to solve a \emph{safety game} by iteratively removing inconsistent states from $T$ until the BTS is consistent.
	Specifically, we define an operator
	\[
	F:\mathcal{T}\to\mathcal{T}
	\]
	by: for any $T$, we have $F(T)=T|_{Q^T_{consist}}$.
	Note that $T|_{Q^T_{consist}}$ need not be consistent in general since removing inconsistent states may
	create new inconsistent states.
	Note that, iterating operator $F$ always converges in a finite number of steps since we need to remove at least one state for each iteration
	and there are only finite number of states in $T$.
	We define
	\[
	T^*=\lim_{k\to\infty}F^k(T_0), \text{ where }T_0=T_{total}|_{Q\setminus Q_{rev}}
	\]
	as the resulting BTS after the convergence of operator $F$.
	
	\begin{remark}
	(Simplification of  Information-States) \upshape  
	In the  information-state updating rule, at each instant, 
	the current state estimate is added to the second component.
	However, 
	if a current state estimate does not even contain a secret state,  
	it is impossible to   infer that the system was at a secret state no matter what is observed in the future. 
	For such a scenario, adding the current state estimate is irrelevant for us to determine $Q_{rev}$. 
	Therefore,   if $C(\imath')\cap X_S=\emptyset$, then the second part of the information-state updating rule in Equation~\eqref{Updating} can be further simplified as
	\begin{equation} 
	D(\imath')
     = \{ \widetilde{U\!R}_{\gamma}(\widetilde{N\!X}_{\sigma}(\rho))\in 2^{X\times X}:\rho\in D(\imath)\}
	\end{equation}
	This simplification does not affect any results of the proposed approach but can reduce the space-state significantly when the number of secret states is relatively small. 
	For the sake of simplicity, we will use this simplified rule in the example by omitting current state estimate in the second component when no secret state is involved. 
	\end{remark}
	
	Before we proceed further, we illustrate the above procedure by the following example.

	\begin{myexm}\upshape
		Let us still consider system $G$ in Fig.~\ref{SystemG}.
		Then $T^*$ in Fig.~\ref{BTS} is a BTS.
		For the sake of simplicity, uncontrollable events $o_1$ and $o_2$ are omitted in each control decision  in Fig.~\ref{BTS}.
		Also, we follow the simplification in the above remark  by  omitting secret-irrelevant information in the second part of $Z$-states.
		At the initial $Y$-state $s_1=(\{0\},\{\emptyset\})$, the supervisor makes control decision $\gamma=\{b,c\}$, we reach $Z$-state $s_5=h_{Y\!Z}(s_1,\gamma)=(\!\{0,\!4,\!5,\!7\}, \{\!\{\!(0,0),\!(0,4),\!(0,5),\!(0,7),$ $\!(4,4),\!(4,5),\!(4,7),
		\!(5,5),\!(7,7)\}\!\},\!\{b,c\}\!)$.
		From $s_5$, the occurrence of observable event $o_2$ leads to the next $Y$-state $s_6=h_{ZY}(s_5,o_2)=(\!\{6,\!8\},\{\!\{\!(0,6),\!(0,8),\!(4,6),\!(4,8),\!(5,6),$ $\!(7,8)\}\}\!)$.
		From $Y$-state $s_6$,
		by making control decision $\gamma=\{a,b,c\}$, 
		the system will reach the next $Z$-state $s_7=h_{Y\!Z}(s_6,\gamma)=(\!\{6,\!8\},\{\!\{\!(0,6),\!(0,8),\!(4,6),\!(4,8),\!(5,6),\!(7,8)\},$ $\!\{(6,6),\!(8,8)\}\!\},\{a,b,c\})$ according the original updating rule in Equation~\eqref{Updating}.
		Note that since the current state estimate is $C(\imath')=\{6,8\}$ in which there is no secret state, 
		we can apply the simplification rule without adding $ \{ \odot_{\gamma}(C(\imath'))\}$. 
		This is why we have 
		 $D(s_7)\!=\!\{\!\{\!(0,6),\!(0,8),\!(4,6),\!(4,8),\!(5,6),\!(7,8)\}\}$ in Fig.~\ref{BTS}. 
	    From $Z$-state $s_7$, observable event $o_1$ occurs and the system will reach the next $Y$-state $s_8=h_{ZY}(s_7,o_1)=(\{6\},\{\{(0,6),(4,6),(7,6)\}\})$.
		Then supervisor makes control decision $\gamma=\{a,b,c\}$ and system will reach to $Z$-state $s_9=h_{Y\!Z}(s_8,\gamma)=(\{6\},\{\{(0,6),(4,6),(7,6)\}\},\{a,b,c\})$. 
		Again, $ \{ \odot_{\gamma}( \{6\}  )\}$ is not added as no secret state is involved. 

		In fact, $T^*$ is obtained as follows. 
		First, we construct the largest BTS that enumerates all possible transitions for each state, which is $T_{total}$ shown in the entire box marked with black lines in Fig.~\ref{BTS}.
		Note that $Q_{rev}=\{s_{18}\}$ since  $D_1(s_{18})=\{\{0,1,2\},\{5\}\}$ and $\{5\}\subseteq X_S$. 
		Therefore  $T_0$ is obtained by restricting $T_{total}$ to $Q\setminus Q_{rev}$, 
		which is shown in the box marked with blue-dashed lines in Fig.~\ref{BTS}.
		Then, we need to remove all inconsistent states from $T_0$.
		Since $s_{18}$ has been removed, $s_{15}$ becomes an inconsistent $Z$-state,
		which should be removed by applying operator $F$ for the first time.
		Then $s_{12}$ becomes a new inconsistent $Z$-state as feasible observable event $o_1$ is not defined;
		hence $s_{12}$ is deleted when applying operator $F$ for the second time.
		We keep applying operator $F$ and need to remove state $s_{11}$.
		This makes $s_{10}$ and $s_{13}$  inconsistent because feasible observable event $o_1$ is not defined; operator $F$ will further remove them.
		Then  operator $F$ converges to $T^*$.
	\end{myexm}
	
	Next, we show that synthesizing a supervisor within   $T^*$ is without loss of generality.
	\begin{mythm}
		The opacity-enforcing synthesis problem has no solution if $T^*$ is empty.
	\end{mythm}
	\begin{proof}  
	(Sketch)
    Suppose that the synthesis problem has a solution $S$.
    Then we know that any information-states reached in $S$ are not in $Q_{rev}$.
    Moreover, for any $Y$ or $Z$-states reached by supervisor $S$ are consistent as $S$ can correctly choose a transition.
    Therefore, all information states reached under $S$ should not be removed during the iteration of operator $F$.
   	Hence, $T^*$ should not be empty if the synthesis  problem has a solution.
    \end{proof}
	
	Then for any $Y$-state $y$ in $T^*$, we define
	\[
	Dec_{T^*}(y)=\{   \gamma\in\Gamma:h_{Y\!Z}^{T^*}(y,\gamma)!        \}
	\]  as the set of control decisions defined at $y $ in $T^*$. 
	Clearly, 
	$\langle Dec_{T^*}(y),\subseteq \rangle$ forms a finite poset, which contains at least one   maximal element $\gamma$ such that  $\forall  \gamma'\in Dec_{T^*}(y): \gamma \not\subset \gamma'$. 
	When $T^*$ is not empty, we can synthesize a supervisor $S^*$ as follows.
	At each instant, the supervisor $S^*$ will remember the current information-state $y$ ($Y$-state)  and
	pick a \emph{maximal} control decision $\gamma$ from $Dec_{T^*}(y)$. \footnote{Some decisions are ``equivalent" at a $Y$-state in the sense that some events in the control decision may not be feasible within the unobservable reach.
	For example, decisions $\{a \}$  and $\{a,c\}$ are equivalent at $s_1$ in Fig.~\ref{BTS} as event $c$ is not feasible.
	For those equivalent decisions, we  only draw the one with all redundant event included in the figure.
	However, when   comparing   decisions to find  local maximal decisions, those redundant events should not be counted.}
	Note that maximal control decision is not unique in general and we denote by $Dec_{T^*}^{max}(y)\subseteq Dec_{T^*}(y)$ the set of maximal control decisions at $y$ in $T^*$.
	Then we update the information-state based on the control decision issued and wait for the next observable event and so forth.
	The execution of $S^*$ is formally described as Algorithm~1.
	\begin{algorithm}
		\SetKwData{Left}{left}
		\SetKwData{Up}{up}
		\SetKwFunction{FindCompress}{FindCompress}
		\SetKwInOut{Input}{input}
		\SetKwInOut{Output}{output}
		\nl $y\gets \{\{x_0\},\{{\emptyset}\}\}$\;
		\nl find a  maximal control decision $\gamma$ from $Dec_{T^*}^{max}(y)$\;
		\nl make initial control decision $\gamma$\;
		\nl $z\gets h_{Y\!Z}^{T^*} (y,\gamma)$\;
		
		\nl \While{new event $\sigma\in \gamma\cap \Sigma_o$ is observed}
		{
			\nl $y\gets h_{ZY}^{T^*} (z,\sigma)$\;
			\nl find a  maximal control decision $\gamma$ from $Dec_{T^*}^{max}(y)$\;
			\nl update the control decision to $\gamma$\;
			\nl $z\gets h_{Y\!Z}^{T^*} (y,\gamma)$\;
		}
		\label{algorithm1}
		\caption{Execution of Qualitative Supervisor $S^*$}
	\end{algorithm}
	
	Finally, we show that the proposed supervisor $S^*$ indeed solves Problem~\ref{problem}.
	
	\begin{mythm}\label{thm2}
		Supervisor $S^*$ defined by Algorithm 1 enforces infinite-step opacity and is maximally permissive.
	\end{mythm}
	\begin{proof} 
    Let $\alpha\in P(\mathcal{L}(S^*/G))$ be any observable string in the closed-loop system $S^*/G$.
    By Corollary~\ref{cor:delay}, we know that $D_1(\mathcal{I}_{S^*/G}(\alpha))=\{\hat{X}_{S^*/G}(\alpha'\mid\alpha):\alpha'\leq\alpha \}$.
    Since all secret-revealing states have been removed from $T_{total}$, all information-states in $T^*$ are safe, which implies that $\hat{X}_{S^*/G}(\alpha'\mid\alpha)\not\subseteq X_s$.
    This means that  $S^*$ enforces infinite-step opacity.

    Next, we show that $S^*$ is maximal.
    Assume that there exists another supervisor $S'$ such that $\mathcal{L}(S^*/G)\subset\mathcal{L}(S'/G)$.
    This means that there exists a string $w\in\mathcal{L}(S^*/G)\subset\mathcal{L}(S'/G)$ such that    $S^*(w)\subset S'(w)$ and $S^*(w')=S'(w')$, $\forall w'< w$.
    Then $Y$-state reached upon the occurrence of string $w$ under supervisor $S'$ and $S^*$ are the same, which is denoted by $y$.
    Since $S'$ is a solution to Problem~\ref{problem}, its control decision $S'(w)$ should not be removed at $y$ during the iteration. This means that   $S'(w) \in Dec_{T^*}(y)$.
    However, it contradicts to our choice that $S^*(w)$ is in $Dec_{T^*}^{max}(y)$.
    Hence no such $S'$ exists.
    \end{proof}
	
    \begin{figure}
    \centering
    \subfigure[Solution $S_1$]{\label{Solution1}
    \includegraphics[width=0.5\textwidth]{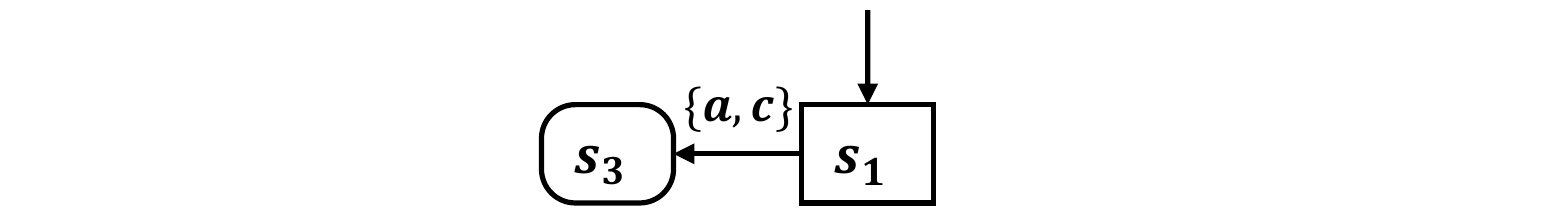}} 
    \subfigure[Solution $S_2$]{\label{Solution2}
    \includegraphics[width=0.5\textwidth]{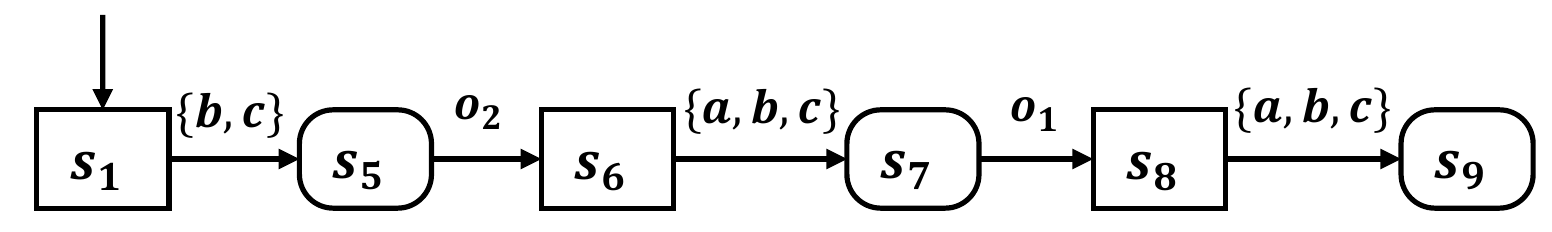}}
    \caption{Solution $S_1$ and $S_2$}
	\label{Solution}
    \end{figure}
	
	\begin{myexm}\upshape
	Again we consider system $G$ in  Fig.~\ref{SystemG} and we use Algorithm~1 to solve Problem 1.
		At the initial $Y$-state $y_0=(\{x_0\},\{\emptyset\})$,
		we have $Dec_{T^*}^{max}(y_0)=\{\{a,c\}, \{b,c\}\}$.
		If we choose control decision $\{a,c\}$, then we reach $Z$-state $s_3$ and the system has no future observation.
		This gives supervisor $S_{1}$ shown in Fig.~\ref{Solution1}, which  results the language generated by $G_1$ in Fig.~\ref{SystemG1}. 
		
		On the other hand, if we choose control decision $\{b,c\}$ initially,
		then the system moves to $s_5$.
		From $s_5$, the occurrence of $o_2$ leads to $s_6$ and the supervisor picks $\{a,b,c\}$ as the maximal control decision and the system moves to $s_7$.
		Then observable events $o_1$ occurs and the system moves to $s_8$, where supervisor $S_{2}$ picks $\{a,b,c\}$ leading the system to state $s_9$.
		This gives supervisor $S_{2}$ shown in    Fig.~\ref{Solution2}, which results in the language generated by $G_1$ in  Fig.~\ref{SystemG2}.
		As we have discussed in Example~\ref{exm:111}, both supervisors are maximal and they are incomparable.
	\end{myexm}
	
	\section{Quantifying  Secret-Revelation-Time}\label{sec:6}
	So far, we have solved a qualitative version of privacy-enforcing control problem by requiring that the closed-loop system under control is infinite-step opaque. 
	In some cases, such a binary requirement may be too strong  as the secret may be revealed inevitably after some delays no matter what control policy is taken, i.e., infinite opacity is not enforcable.

	In most of the applications, however, the importance of secret  will decrease as time goes on. 
	Then for the scenario where infinite-step opacity is not enforcable, it makes senses to consider an optimal synthesis problem by maximizing the \emph{secret-revelation-time} for each visit of secret states. 
	In the remaining part of the paper, we will implement this idea by further generalizing the qualitative synthesis problem to a quantitative version by ensuring the secret be revealed \emph{as late as possible}. 
	 
	\subsection{Problem Formulation of  Quantitative Synthesis Problem}  
	Let $S$ be a supervisor, 
	$\alpha\in P(\mathcal{L}(S/G))$ be an observation and $\alpha'\leq \alpha$ be a prefix of $\alpha$. 
	We define
	\[
	\textsc{Rev}(\alpha', \alpha)
	=\{\beta \leq \alpha :\hat{X}_{S/G}(\alpha'\mid \beta)\subseteq X_S \} 
	\]
	as the set of observations that are of prefixes of $\alpha$ and suffixes of $\alpha'$, upon which the visit of secret state at instant $\alpha'$ is revealed. 
	If secret state is visited at instant $\alpha'$ and it is revealed when $\alpha$ is executed, 
	then we have $\textsc{Rev}(\alpha', \alpha)\neq \emptyset$.  
	Note that once we know that the system was at secret state at instant $\alpha'$, we know this forever.  
	Therefore, we are interested in the \emph{first instant} when the secret is revealed, i.e., 
	the shortest string	$ \beta_{short}  \leq \alpha$  in  $\textsc{Rev}(\alpha', \alpha) $. 
	Then $|\beta_{short}|-|\alpha'|$ is referred to as the \emph{secret-revelation-time} for the instant $\alpha'$ upon $\alpha$.

	To quantify when the visit of a  secret state is revealed, 
	we consider a cost function 
	\[
	C: \mathbb{N}\rightarrow \{0,1,\dots,N_{max}\}
	\]
	that assigns each secret-revelation-time to a non-negative integer cost. 
	We assume that the cost function $C$ is monotonically non-increasing,  
	i.e., the importance of secret decreases as time goes on and it is more desirable to reveal the secret as late as possible.

	To avoid counting the revelation of each visit of secret state duplicatively, we define the cost incurred for $\alpha'$ upon $\alpha$ as the  cost incurred at the \emph{first} secret-revelation-instant  as we assume the cost function is non-increasing, 
	i.e.,  
	\begin{equation}
	    \texttt{Cost}(\alpha', \alpha)\!=\!\!\! 
	    \left\{\!
	    \begin{aligned}
		&0 &\text{if }& \textsc{Rev}(\alpha', \alpha)\!=\!\emptyset\\
		&  \max_{\beta\in \textsc{Rev}(\alpha', \alpha)} C(|\beta|-|\alpha'\mid)
		&\text{if }& \textsc{Rev}(\alpha', \alpha)\!\neq\!\emptyset
		\end{aligned}
		\right.\nonumber 
	\end{equation}	 
	Note that,  for string $st\in \mathcal{L}(S/G)$, if   $\delta(s)\notin X_S$, then we always have 
	$\texttt{Cost}(P(s),P(st))=0$.  
	
	Finally, we define the total cost incurred along an observation. 
  	Note that there may have multiple visits of secret states along string $\alpha$ at different instants. 
  	Therefore, we define the  \emph{total cost} incurred upon $\alpha\in P(\mathcal{L}(S/G))$ 
  	as the summation of the costs for all visits of secret states and their associated secret-revelation-times, i.e., 
	\[
	\texttt{Cost}(\alpha)=\sum_{\alpha'\leq\alpha } \texttt{Cost}(\alpha', \alpha).
	\]
	Then the cost of the closed-loop system $S/G$  is defined as the worst-case cost among all possible strings, i.e., 
    \[
	\texttt{Cost}(S/G)=\max\limits_{\alpha\in P(\mathcal{L}(S/G))} \texttt{Cost}(\alpha).
	\]
	Then the quantitative privacy-enforcing synthesis problem is then formulated as follows. 
	\begin{myprob}(Quantitative Privacy-Enforcing Control Problem)\label{problem2}
	Given a system $G$ and a set of secret states $X_S\subseteq X$, 
	determine whether or not there exists  a supervisor $S:P(\mathcal{L}(G))\rightarrow\Gamma$ such that its cost is finite. If so,   synthesize an optimal supervisor $S$ such that 
	\begin{enumerate}[(i)]
	\item
	for any  $S'$,  we have
	$\texttt{Cost}(S/G)\leq \texttt{Cost}(S'/G)$; and 
	\item 
	for any $S'$ such that $\texttt{Cost}(S/G)= \texttt{Cost}(S'/G)$, 
	we have    $\mathcal{L}(S/G)\not\subset\mathcal{L}(S'/G)$.
	\end{enumerate}
	\end{myprob}
	
	\begin{remark}\upshape
	The quantitative formulation   using cost function generalizes the notion of infinite-step opacity investigated above, 
	as well as the notions of  current-state opacity \cite{wu2013comparative} and $K$-step opacity \cite{saboori2011verification} in the literature, which requires  that the visit of a secret state should not be revealed currently and in the next $K$ steps, respectively. 
	Specifically, for infinite-step opacity, it suffices to consider  cost function $C_{Inf}$ defined by 
	\[\forall k\geq 0: C_{Inf}(k)=\infty.\] 
	For $K$-step opacity, it suffices to consider   cost function $C_{Kst}$ defined by 
	\[\forall k\leq  K: C_{Kst}(k)=\infty \text{ and }\forall k>  K: C_{Kst}(k)=0,\]
	and current-state opacity is nothing but $0$-step opacity. 
	These existing notions are essentially binary, while our new formulation allows to investigate the effect of secret revelation delays more quantitatively. 
	\end{remark}
	
	\begin{remark} \upshape
	For the sake of simplicity, hereafter, we consider a cost function $C$ in a simple specific form of
	\[
	C(k)=\max\{N_{max}-k,0\}, 
	\]
	where $N_{max}$ is a finite value. 
	That is,  $C(k)$  is the cost incurred if the visit of a state state is revealed for the first time after $k$ steps.
	Our approach, in principal, can be applied to any form of cost function as long as it decreases to zero in a finite number of steps.
	\end{remark}

	\begin{figure}
    \centering
    \includegraphics[width=0.3\textwidth]{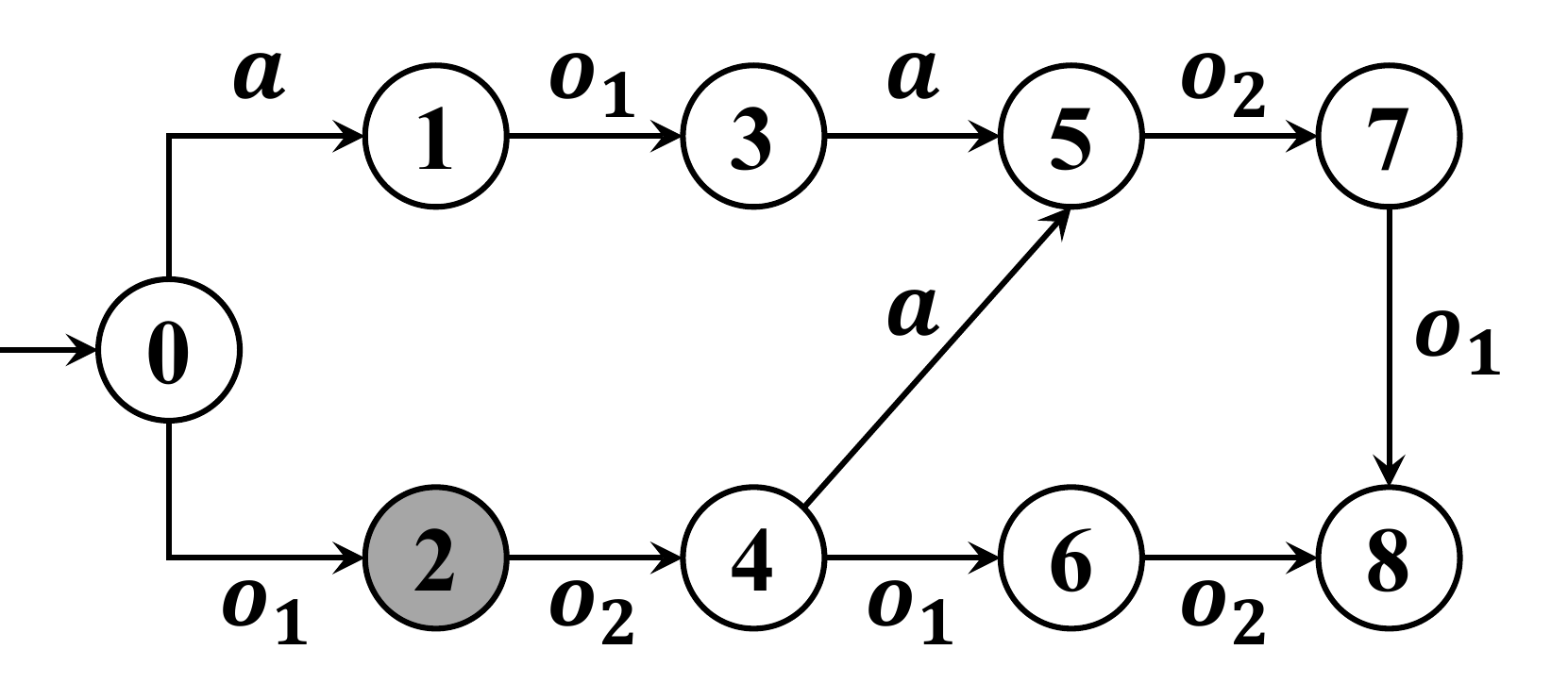} 
    \caption{System $G$ with $\Sigma_o=\{o_1,o_2\},\Sigma_c=\{a\},X_s=\{2\}$.}
	\label{SystemG3}
    \end{figure}
	
	\begin{myexm}\upshape 
	Let us consider system $G$ in  Fig.~\ref{SystemG3} with $\Sigma_o=\Sigma_{uc}=\{o_1, o_2\}$ and $X_S=\{2\}$. 
	Note that all events in string  $o_1o_2o_1o_2$ are uncontrollable and there does not exist another string in $G$ having the same projection. Therefore,  we cannot enforce infinite-step opacity qualitatively, 
	i.e., no matter what the supervisor does, the intruder will know for sure that the system is/was at state $2$ for the instant when $o_1$ is observed. However, the supervisor can control how late the secret is revealed. 
	
	For example, let us consider supervisor $S_1$ that always enables all events. 
	Then for $P(o_1o_2o_1o_2)=o_1o_2o_1o_2$, we have $ \textsc{Rev}(o_1, o_1o_2o_1o_2)=\{o_1o_2o_1o_2\}$ 
	and $ \textsc{Rev}(\alpha', o_1o_2o_1o_2)=\emptyset$ for $\alpha'\neq o_1$. 
	Therefore, we have 
	\begin{align}
	&\texttt{Cost}(o_1o_2o_1o_2)
	=\sum_{\alpha'\leq o_1o_2o_1o_2 } \texttt{Cost}(\alpha', o_1o_2o_1o_2)\nonumber\\
	&= \texttt{Cost}(o_1, o_1o_2o_1o_2) = C(|o_1o_2o_1o_2|-|o_1|) =N_{max}-3\nonumber
	\end{align}
	Similarly, we can compute 
	$\texttt{Cost}(o_1o_2o_2o_1)=\texttt{Cost}(o_1, o_1o_2o_2o_1)=C(|o_1o_2o_2|-|o_1|)= N_{max}-2$, 
	which is the worst-case cost in $S_1/G$, i.e., 
	$\texttt{Cost}(S_1/G)=N_{max}-2$.
 
    If we consider supervisor $S_2$ that disables $a$ initially, 
    then we have $\textsc{Rev}(o_1, o_1o_2o_1o_2)=\{o_1,o_1o_2,o_1o_2o_1,o_1o_2o_1o_2\}$
    and $\textsc{Rev}(o_1, o_1o_2o_2o_1)=\{o_1,o_1o_2,o_1o_2o_2,o_1o_2o_1o_1\}$. 
    Therefore, 
    $\texttt{Cost}(o_1o_2o_1o_1)\!=\!\texttt{Cost}(o_1o_2o_2o_1)\!=\!\texttt{Cost}(o_1)\!= \!N_{max}$, 
    which is the worst-case cost in $S_2/G$, i.e., 
	$\texttt{Cost}(S_2/G)=N_{max}$.
	\end{myexm}

	\subsection{Augmented Information-State} 
	To solve the quantitative synthesis problem, the previous proposed information-state is not sufficient as the time information is lost, i.e., we only remember all possible delayed state estimates without specifying \emph{when} they are visited. 
	Therefore, we further augment this information to the previous proposed information-states, which  leads to the \emph{augmented information-states} defined by 
	\[
	I_a:=2^X\times 2^{2^{X\times X}\times \{0,1,\ldots, N_{max}-1\}}.
	\]
	Each augmented information-state $\imath \in I_a$ is in the form of $\imath=(C_a(\imath),D_a(\imath))$. 
	The first component $C_a(\imath)\in 2^X$ is still a current state estimate. 
	Each element in the second component $ D_a(\imath)$ is in the form of  
	\[
	(m,k)=(\{(x_1,x'_1),(x_2,x'_2),\cdots,(x_n,x'_n)\},k)\in D_a(\imath), 
	\]
	where  the first and the second parts represent, respectively, a delayed state estimate and when it is visited, 
	and we define 
	\[
	m_1=\{x_1,x_2,\dots, x_n\}=\{x\in X: (x,x')\in m\}.
	\]
	Note that we only augment the timing information for $N_{max}$ steps as we consider a cost function that decreases to zero in $N_{max}$ steps. 
    For an augmented information-state $\imath \in I_a$, we define
    \begin{equation}
    \textsc{Rev}(\imath)=\{(m,k)\in D_a(\imath):  m_1\subseteq X_S \} 
    \end{equation}
    as those secret-revealing delayed state estimates.

	Similarly, suppose that the system's  augmented information-state is $\imath=(C_a(\imath),D_a(\imath))$.
	Then, upon the occurrence of an observable event  $\sigma\in \Sigma_{o}$ and a new  control decision $\gamma\in\Gamma$,  the augmented information-state  $\imath$ is updated  to $\imath'=(C_a(\imath'),D_a(\imath'))$ as follows:

	\begin{center}
        \tcbset{width=8.5cm,lefttitle=0.7cm,title=Augmented Information-State Updating Rule,colback=white}
    \begin{tcolorbox} 
\begin{equation}
	\!\!\!\!\!\left\{
	\begin{aligned}\label{Updating1}
	C_a(\imath')
	= &U\!R_{\gamma}(N\!X_{\sigma}(C_a(\imath)))\\
	D_a(\imath')
	= &
	 \!\left\{\!\! \!
	 \begin{array}{c c}
	 (\widetilde{U\!R}_{\gamma}(\widetilde{N\!X}_{\sigma}(m)),k+1): \\
	 (m, k) \in D_a(\imath)\setminus \textsc{Rev}(\imath),  k+1< N_{max}
	\end{array}\!\!
	\right\} \\
	& \cup \{( \odot_{\gamma}(C_a(\imath')),0)\} 
	\end{aligned}
	\right.
	\end{equation}
    \end{tcolorbox}
\end{center}
	
	Compared with the information-state updating rule  in Equation~\eqref{Updating}, 
	the augmented information-state updating rule in Equation~\eqref{Updating1} has the following differences: 
	\begin{enumerate}
	    \item 
	    The timing information is also tracked, which increases a time unit upon the occurrence of each observable event; and
	    \item
	    Only those delayed state estimates in the last $N_{max}-2$ steps and  have not yet been revealed are updated.  
	    This is because  the cost decreases to zero within $N_{max}$ steps and we only consider the cost incurred for the first secret-revelation-instant. 
	\end{enumerate} 
 	
    Still, let $\alpha=\sigma_1\sigma_2\cdots\sigma_n\in P(\mathcal{L}(S/G))$ be an observable string.
    We denote by    $\mathcal{I}_a(\alpha)$ the augmented information-state reached by $\alpha$, which is defined according to Equation~\eqref{In-Evo1} 
    using the augmented updating rule in Equation~\eqref{Updating1} with initial state  
    $\imath_0=(U\!R_{S(\epsilon)}(\{x_0\} ),\{(\odot_{S(\epsilon)}(U\!R_{S(\epsilon)}( \{x_0\} )), 0)\})$.
    Similar to the previous information-state, the augmented information-state has the following properties.  
    
	\begin{mypro}\label{pro:2}
	Let $S$ be a supervisor, $\alpha\in P(\mathcal{L}(S/G))$ be an observable string and $\mathcal{I}_a(\alpha)$ be the augmented information-state reached.
	Then we have
		\begin{enumerate}[(i)]
			\item
			$C_a(\mathcal{I}_a(\alpha))=\hat{X}_{S/G}(\alpha)$; and
			\item
			$D_a(\mathcal{I}_a(\alpha))=\\
			\left\{
			(\rho_{\alpha',\alpha}, |\alpha\mid-|\alpha'\mid)
			:
			\begin{gathered}
            \alpha'\leq\alpha , |\alpha\mid-|\alpha'\mid<  N_{max},\\
			[\forall \alpha'\leq \beta< \alpha :  \textsc{Rev}(\alpha',\beta)=\emptyset]
			\end{gathered}
			\right\}$ 
		\end{enumerate} 
	\end{mypro}
	\begin{proof}
	The proof is very similar to that of Proposition~\ref{prop:main}; hence a detailed proof is omitted.  
	The only differences are: 
	(i) the augmented updating rule has a counter that remembers the number of steps between $\alpha'$ and $\alpha$; and 
	(ii) only those delayed state estimates that are not secret-revealing are updated, 
	i.e., $ \textsc{Rev}(\alpha',\beta)$ should be empty-set for any strict prefix $\beta$ before $\alpha$; 
	otherwise, the delayed state estimate will be dropped according to the updating rule.  
	\end{proof}
	
	With the help of Proposition~\ref{pro:2}, now  we can relate the proposed augmented information-state with the cost function as follows. 
	Note that a secret-revelation  cost occurs at the instant when the secret is revealed \emph{for the first time}.   
	This is captured by $\textsc{Rev}(\mathcal{I}_a(\alpha))$
	and for any $(m,k)\in \textsc{Rev}(\mathcal{I}_a(\alpha))$ such that $m_1\subseteq X_S$, this secret revelation is only counted once as it will not be updated according to the augmented updating rule.  
	Therefore, we can define a state-based cost function 
	\[C_I: I_a\to \{0,1,\dots,N'_{max}\}
	\]
	assigning each augmented information-state $\imath\in I_a$ a cost by 
	\begin{equation}
	    C_I(\imath)= \sum_{(m,k)\in \textsc{Rev}(\imath)}  C(k)
	\end{equation}
	Note that the cost of each information-state is upper-bounded by $N'_{max}\leq 1+2+\cdots+N_{max}=\frac{1}{2}N_{max}(N_{max}+1)$, 
	which corresponds to the extreme case where each of the previous $N_{max}$ visits a secret state and they are all revealed for the first time at the current instant.   
	
    The following result shows that, for any  observable string $\alpha\in P(\mathcal{L}(S/G))$, 	
    the total secret-revelation cost incurred $\texttt{Cost}(\alpha)$ is equal to the summation of the costs of all information-states reached along $\alpha$. 
    
	\begin{mythm}\label{thm:cost}
	Let $S$ be a supervisor, $\alpha\in P(\mathcal{L}(S/G))$ be an observable string and 
	for each prefix $\alpha'\leq\alpha$,  $\mathcal{I}_a(\alpha')$ be the augmented information-state reached by $\alpha'$.
	Then we have
    \begin{equation}
    \texttt{Cost}(\alpha)= \sum_{\alpha'\leq\alpha }C_I(\mathcal{I}_a(\alpha')).
    \end{equation}
	\end{mythm}	
	\begin{proof} 
	Suppose $\alpha=\sigma_1\sigma_2\ldots\sigma_n\in P(\mathcal{L}(S/G))$ be an observable string.
	Let $i_1< i_2<\cdots< i_k$ be the indices such that $\textsc{Rev}(\sigma_1\ldots\sigma_{i_p}, \alpha)\neq\emptyset, \forall p=1, \ldots, k$. 
	Then for each $i_p$, we denote by $j_p$ the first instant when the secret at instant $\sigma_1\ldots\sigma_{i_p}$  revealed, 
	i.e., $\sigma_1\ldots\sigma_{j_p}$ is the shortest string in $\textsc{Rev}(\sigma_1\ldots\sigma_{i_p}, \alpha)$. 
	Then by the definition of $\texttt{Cost}(\alpha)$, we have
	\[\texttt{Cost}(\alpha)
	=\sum_{\alpha'\leq\alpha } \texttt{Cost}(\alpha', \alpha)
	=\sum\limits_{p=1,\ldots,k}C(j_p-i_p)\]
	Then let $\imath_0,\imath_1,\dots,\imath_n$ be all augmented information-states reached along $\alpha$. 
	Then by Equation~\eqref{Updating1} and Proposition~\ref{pro:2}, 
	each $p=1,\dots, n$ only contributes a cost of $C(j_p-i_p)$ via $C_I$ at information-state $\imath_{j_p}$. 
	Therefore, we have 
	\begin{equation}
	\begin{aligned}
	&\sum_{\alpha'\leq\alpha }C_I(\mathcal{I}_a(\alpha'))  = \sum_{\alpha'\leq\alpha }\sum_{(m,k)\in \textsc{Rev}(\mathcal{I}_a(\alpha'))} C(k) \\
	=& \sum_{ i=1,\dots,n }C_I( \imath_i )  
	=\sum\limits_{p=1,\ldots,k}C(j_p-i_p)\nonumber
	\end{aligned}
	\end{equation}
	This completes the proof.
	\end{proof}
    
	\begin{myexm}\upshape
	Let us consider system $G$ in  Fig.~\ref{SystemG3} and $N_{max}=5$. 
	Suppose that  supervisor $S$ enables $a$ initially, i.e. $S(\epsilon)=\{o_1, o_2, a\}$.
	Then we have
	$$
	\begin{aligned}
	    \mathcal{I}_a(\epsilon)=&(U\!R_{S(\epsilon)}(\{x_0\} ),\{(\odot_{S(\epsilon)}(U\!R_{S(\epsilon)}( \{x_0\} )), 0)\})\\
	=&(\{0,1\}, \{(\{(0,0), (0,1), (1, 1)\}, 0)\})
	\end{aligned}
	$$
    When event $o_1$ is observed, if the decision of the supervisor is $S(o_1)=\{o_1, o_2\}$, 
    then the augmented information-state is updated to $\mathcal{I}_a(o_1)$, where
	$$
		\begin{aligned}
		C_a(\mathcal{I}_a(o_1))=&U\!R_{S(o_1)}(N\!X_{o_1}(C_a(\mathcal{I}_a(\epsilon))))  =\{2, 3\}\\
		D_a(\mathcal{I}_a(o_2))=&\left\{\!\!\! 
	    \begin{array}{c c}
	    (\widetilde{U\!R}_{S(o_1)}(\widetilde{N\!X}_{o_1}(m)),k+1): \\
	    (m,k) \in D_a(\mathcal{I}_a(\epsilon))\setminus \textsc{Rev}(\mathcal{I}_a(\epsilon)),  k+1< 5
	    \end{array}\!\!\!\right\} \\
	    & \cup \{( \odot_{S(o_1)}(C_a(\mathcal{I}_a(o_1))),0)\} \\
		=&
		\left\{\!\!\!
		\begin{array}{c c}
		(\{(0,2),(0,3),(1,3)\}, 1),\\
		(\{(2,2),(3,3)\}, 0)
		\end{array}
		\!\!\!\!
		\right\}
		\end{aligned}
		$$
		Again, if the supervisor further makes  control decision $S(o_1o_2)=\{o_1, o_2, a\}$ 
		when event $o_2$ is observed, then 
		the augmented information-state is updated to $\mathcal{I}_a(o_1o_2)$, where
		$$
		\begin{aligned}
		&C_a(\mathcal{I}_a(o_1o_2))=U\!R_{S(o_1o_2)}(N\!X_{o_2}(C_a(\mathcal{I}_a(o_1))))
		=\{4, 5\}\\
		&D_a(\mathcal{I}_a(o_1o_2))\\
		=
		&\left\{\!\!\! \!
	    \begin{array}{c c}
	    (\widetilde{U\!R}_{S(o_1o_2)}(\widetilde{N\!X}_{o_2}(m)),k+1): \\
	    (m,k) \in D_a(\mathcal{I}_a(o_1))\!\setminus\! \textsc{Rev}(\mathcal{I}_a(o_1))),  k+1< 5
	    \end{array}\!\!\!\!\right\} \\
	     &\cup \{( \odot_{S(o_1o_2)}(C_a(\mathcal{I}_a(o_1))),0)\} \\
		=
		&\left\{\!\!\!
		\begin{array}{c c}
		(\{(0,4),(0,5)\}, 2), 
		(\{(2,4),(2,5)\}, 1), \\
		(\{(4,4),(4,5),(5,5)\}, 0)
		\end{array}
		\!\!\!\!
		\right\}
		\end{aligned}
		$$
		Note that there is a secret-revealing delayed state estimate in $\mathcal{I}_a(o_1o_2)$.
		Specifically, for $(m,k)=(\{(2,4), (2,5)\}, 1)\in D_a(\mathcal{I}_a(o_1o_2))$, we have $m_1=\{2\}\subseteq X_S$.
		Therefore,  we have $\textsc{Rev}(\mathcal{I}_a(o_1o_2))=\{(\{(2,4), (2,5)\}, 1)\}$ and 
		the cost of augmented information-state is $C_I(\mathcal{I}_a(o_1o_2))=C(1)=N_{max}-1=4$.
	\end{myexm}
	\subsection{Quantitative Synthesis Algorithm}
	
	Theorem~\ref{thm:cost} suggests the basic idea for solving the quantitative synthesis problem.  
	One can consider the privacy-enforcing control problem as an optimal control problem for \emph{accumulated total cost}.  
	Similar to the safety control problem over the information-state space for the qualitative synthesis, 
	we can solve the quantitative optimal control problem over the augmented-state-space defined by the \emph{augmented bipartite transition system} (A-BTS), which is the same of the BTS but incorporating the augmented information-state updating rule. 
	
	\begin{mydef}
		An augmented bipartite transition system (A-BTS) $T$ w.r.t.\ $G$ is a 7-tuple.
		\begin{equation}\label{T1}
		T=(Q_Y^T,Q_Z^T,h_{Y\!Z}^T,h_{ZY}^T,\Sigma_o,\Gamma,y_0^T),
		\end{equation}
		where
		\begin{itemize}
			\item
			$Q_Y^T\subseteq I_a$ is the set of $Y$-states.
			Therefore, a $Y$-state $y\in Q_Y^T$ is in the form of  $y=(C_a(y),D_a(y))$;
			\item
			$Q_Z^T\subseteq I_a\times\Gamma$ is the set of $Z$-states.
			For each $z\in Q_Z^T$,  $I_a(z)$ and  $\Gamma(z)$ denote, respectively, the augmented information-state and the control decision, so that $z=(I_a(z),\Gamma(z))$.
			For simplicity, we  write  $z=(C_a(z), D_a(z), \Gamma(z))$;
			\item
			$h_{Y\!Z}^T:Q_Y^T\times\Gamma\rightarrow Q_Z^T$ is the partial transition function from $Y$-states to $Z$-states satisfying the following constraint:  for any $h_{Y\!Z}^T(y,\gamma)=z$, we have
			\begin{equation}\label{hYZ2}
			\left\{
			\begin{aligned}
			C_a(z)=& U\!R_{\gamma}(C_a(y)) \\
			D_a(z)=& \{ (\widetilde{U\!R}_{\gamma}(m),k):
			(m,k)\!\in\! D_a(y)\}
			\\&\cup \{ (\odot_{\gamma}(C_a(z)),0)\}  \\
			\Gamma(z)=& \gamma
			\end{aligned}
			\right.
			\end{equation}
			\item
			$h_{ZY}^T:Q_Z^T\times\Sigma\rightarrow Q_Y^T$ is the partial transition function from $Z$-states to $Y$-states satisfying the following constraint:  for any  $h_{ZY}^T(z,\sigma)=y$, we have $\sigma\in\Gamma(z)\cap\Sigma_o$ and
			\begin{equation}\label{hZY2}
			\left\{
			\begin{aligned}
			C_a(y)=&N\!X_\sigma(C_a(z)) \\
			D_a(y)=&
			\left\{\!\!
			\begin{array}{c c}
			(\widetilde{N\!X}_\sigma(m),k+1):\\
			(m,k)\!\in\! D_a(z)\!\setminus\! \textsc{Rev}(z), k+1\!< \!N_{max}
			\end{array}\!\!
			\right\}
			\end{aligned}
			\right.
			\end{equation}
			\item
			$\Sigma_o$ is the set of observable events of $G$;
			\item
			$\Gamma$ is the set of admissible control decisions of $G$;
			\item
		    $y_0^T:=(\{x_0\},\{\emptyset\})\in Q_Y^T$ is the initial $Y$-state.
		\end{itemize}
	\end{mydef}
	
 \begin{figure*}[htbp]
		\centering
		\includegraphics[width=0.8\linewidth]{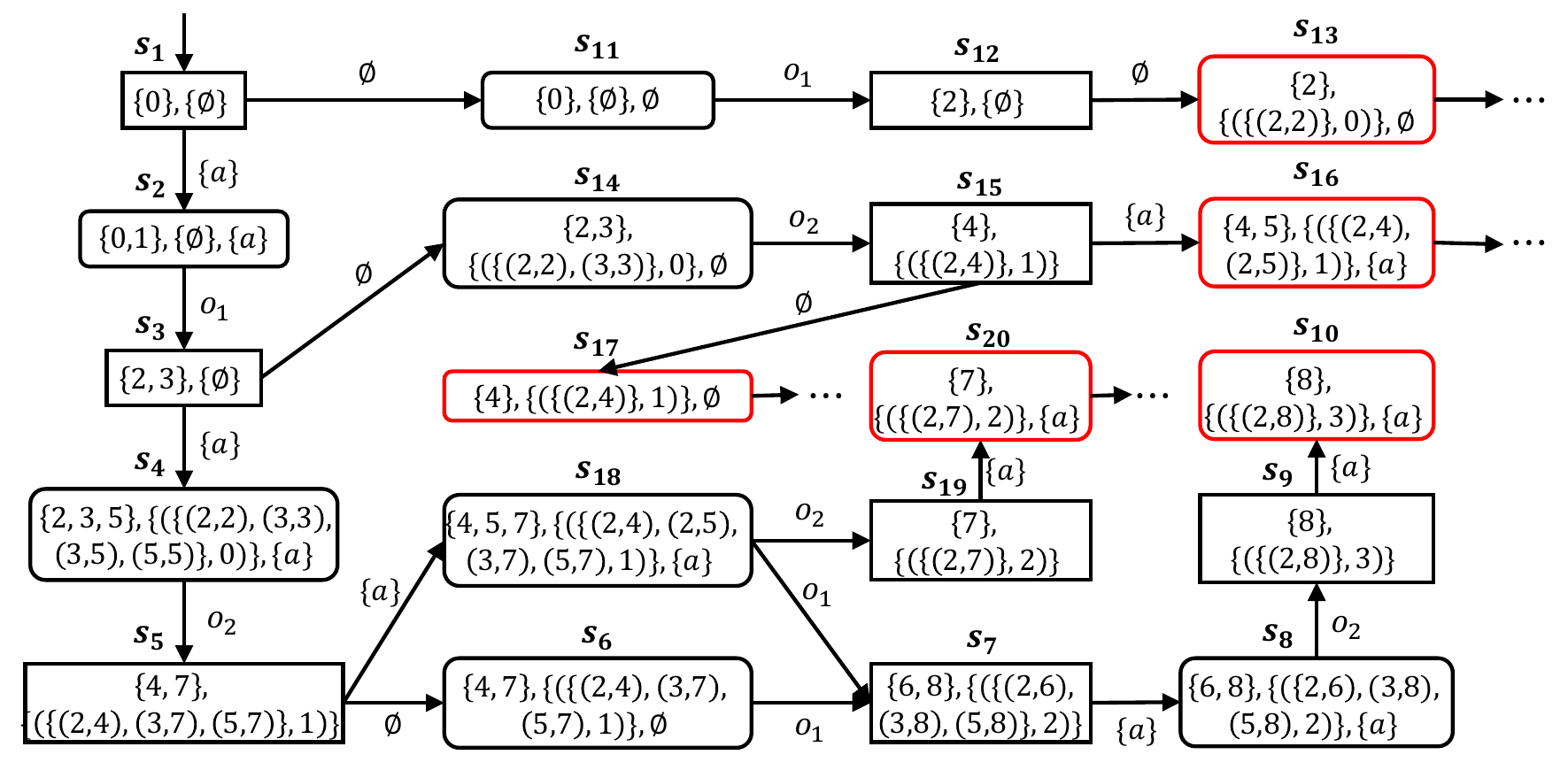}
		\caption{Example of the construction of the A-BTS $T_{total}$.
			Rectangular states correspond to $Y$-states and rectangular  states with rounded corners correspond to $Z$-states.}
		\label{BTS2}
	\end{figure*}
	
    To solve the quantitative synthesis problem, one can still think it as a two-player game between the supervisor and the environment. 
    The goal of the supervisor is to minimize the total cost incurred, while the environment wants to maximize the cost. 
    Still, we  construct the largest A-BTS w.r.t. $G$ that enumerates all the feasible transitions satisfying the constraints of $h_{Y\!Z}^T$ and $h_{ZY}^T$ and denote such an all-feasible A-BTS by $T_{total}$, which is the arena of the game.   
    For each  state $q \in Q_Y^{T_{total}}\cup Q_Z^{T_{total}}$,  we denote by $\textsc{Post}(q)$ the set of all its successor  states.  
    Then we compute the value of each state in $T_{total}$ by  value   iterations as follows: 
	\begin{equation}
	V_{k+1}(q)\!=\! \left\{\!
		\begin{aligned}
		&\min_{q'\in \textsc{Post}(q)}{V_k(q')}  &\text{if }& q\in Q_Y^{T_{total}}\\
		& \max_{q'\in \textsc{Post}(q)}{V_k(q')}+ C_I( \mathcal{I}_a(q)) &\text{if }& q\in Q_Z^{T_{total}}
		\end{aligned}
		\right.
		\label{ValueIteration}
	\end{equation}
    and the initial value function is 
    \[
    \forall q\in Q_Y^{T_{total}}\cup Q_Z^{T_{total}}: V_0(q)=0.
    \]  
    Note that, by our construction, the secret-revelation cost occurs at each $Z$-state $z$ such that $\textsc{Rev}(I_a(z))\neq \emptyset$; 
    this is why the current cost is only added to $Z$-states at each iteration. 
    
    The above is the standard value iteration technique that has been extensively investigated in the literature, 
    either in the context of optimal total-cost control problem \cite{puterman2014markov,bertsekas1995dynamic} or in the context of resource games \cite{chakrabarti2003resource,boker2014temporal}. 
    The value iteration will converge to the value function denoted by $V^*$, possibly in infinite number of steps as the value of some states may be infinite.   
    However, it is known that such a value for each state can be determined only by a finite number of iterations for at most $L=n^2\cdot N'_{max}$ steps \cite{wu2016synthesis}, where $n=|Q_Y^{T_{total}}\cup Q_Z^{T_{total}}|$ is the number of states in $T_{total}$.  
    Specifically, by computing value function $V_L$, we have   
    \begin{equation}
	V^*(q)\!=\! \left\{\!
    \begin{array}{l l}
       V_L(q) \quad
       &\text{if } V_L(q)< n\cdot N'_{max} \\
       \infty\quad
       &\text{otherwise}
     \end{array}
		\right. .
	\end{equation} 
    In other words, $V^*(q)$ is the best cost-to-go the supervisor can guarantee at state $q$.  
    
    Based on the above discussion, Algorithm~2 is proposed to solve Problem~2. 
    First, it builds $T_{total}$ based on the A-BTS and computes the value function $V^*$. 
    If $V^*(y_0)=\infty$, then we cannot find a supervisor whose cost is finite. 
    Otherwise, $V^*(y_0)$ is the optimal cost one can achieve, i.e., $\texttt{Cost}(S/G)=V^*(y_0)$. 
    To execute the supervisor, we use a variable  $C_{rem}$ to record the total cost remained for the supervisor  to attain the optimal value. 
    Formally, let $V^*$ be the value function and $y$ be a $Y$-state in $T_{total} $ and $C_{rem}$ be the cost remaining. 
    We define 
    \[
     Dec_{V^*}(y,C_{rem})=\{ \gamma\in \Gamma:   z=h_{Y\!Z}^{T_{total}}(y,\gamma), V^*(z)\leq C_{rem} \}
    \]
    as the set of all control decisions that attain the optimal value.  
    Clearly, $\langle Dec_{V^*}(y,C_{rem}),\subseteq \rangle$ is also a finite poset and 
    we also denote by $Dec_{V^*}^{max}(y,C_{rem})$ the set of all maximal elements in $Dec_{V^*}(y,C_{rem})$.
    Then at each  $Y$-state, the supervisor chooses a maximal control decision from $Dec_{V^*}^{max}(y,C_{rem})$. 
    Once a  $Z$-state is reached, the cost remained is updated to $C_{rem}-C_I(\mathcal{I}_a(z))$ as a state cost incurred.

	\begin{algorithm}
		\SetKwData{Left}{left}
		\SetKwData{Up}{up}
		\SetKwFunction{FindCompress}{FindCompress}
		\SetKwInOut{Input}{input}
		\SetKwInOut{Output}{output}
		\nl $y\gets \{\{x_0\},\{{\emptyset}\}\}$, $C_{rem}=V^*(y)$\;
		\nl find a   control decision $\gamma$ from $Dec_{V^*}^{max}(y,C_{rem})$\;
		\nl make initial control decision $\gamma$\;
		\nl $z\gets h_{Y\!Z} (y,\gamma)$, $C_{rem}\gets C_{rem}-C_I(\mathcal{I}_a(z))$\;
		
		\nl \While{new event $\sigma\in \gamma\cap \Sigma_o$ is observed}
		{
			\nl $y\gets h_{ZY} (z,\sigma)$\;

			\nl find a   control decision $\gamma$ from $Dec_{V^*}^{max}(y,C_{rem})$\;
			\nl update the control decision to $\gamma$\;
			\nl $z\gets h_{Y\!Z} (y,\gamma)$,  $C_{rem}\gets C_{rem}-C_I(\mathcal{I}_a(z))$\;
		}
		\label{algorithm2}
		\caption{Execution of Quantitative Supervisor $S^*$}
	\end{algorithm}
	
	\begin{mythm}
	Supervisor $S^*$ defined by Algorithm~2 is optimal  and maximally permissive among all optimal supervisors.
	\end{mythm}
	\begin{proof}
	(Sketch) 
	Suppose that there exists a supervisor $S'$ such that $\texttt{Cost}(S'/G)< \texttt{Cost}(S^*/G)$. 
	Then the value of the initial $Y$-state should be at least $\texttt{Cost}(S'/G)$, which is a contradiction to the result of the value iteration. 
    The proof of maximal permissiveness  is similar to that of Theorem \ref{thm2}. 
	Since we choose a control decision from $Dec_{V^*}^{max}(y,C_{rem})$, 
	any other more permissive choices will result in a supervisor whose worst-case cost is larger than $\texttt{Cost}(S^*/G)$, 
	which will violate the optimality.
	\end{proof}
	
	 \begin{figure}
    \centering
    \includegraphics[width=0.35\textwidth]{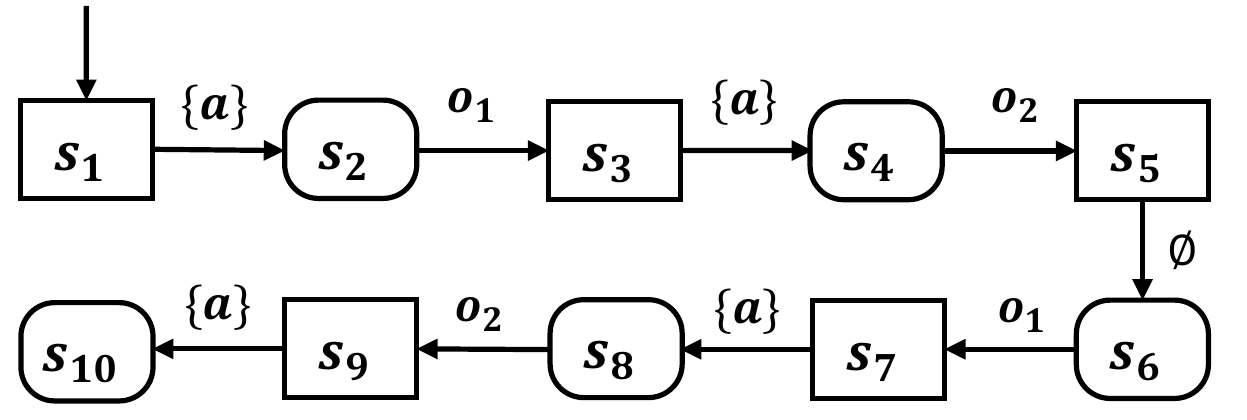} 
    \caption{Solution $S$.}
	\label{Solution3}
    \end{figure}
    
    \begin{remark}\upshape
    Still, in the augmented information-state updating rule, 
    when  $C_a(\imath')\cap X_S=\emptyset$, we do not really need to add $\{( \odot_{\gamma}(C_a(\imath')),0)\}$ to $D_a(\imath')$. 
    This is because the estimate of such an instant will never contribute to the cost function $C_I$ no matter what is observed in the future. For the sake of clarity, we  used the original completed rule in Equation~\eqref{Updating1} for the theoretical developments. However, for the sake of simplicity, we will adopt this simplification in the following illustrative  example. The reader should be aware of this discrepancy. 
    \end{remark}

    Finally, we illustrate how to synthesize an optimal quantitative supervisor by the following example. 

	\begin{myexm}\upshape
    Let us consider system $G$ in  Fig.~\ref{SystemG3}
    and our goal is to synthesize an optimal  supervisor such the secret-revelation cost of the closed-loop system is minimized.
    Suppose that  $N_{max}=5$. 
    First, we construct the largest A-BTS $T_{total}$ that enumerates all possible transitions, which partially shown in  Fig.~\ref{BTS2}.
    For each $Z$-state in $T_{total}$, we find those secret-revealing delayed state estimates and assign each of them a state cost.
    Specifically, we have $C_I(s_{10})=2$, $C_I(s_{13})=5$, $C_I(s_{16})=4$, $C_I(s_{17})=4$, $C_I(s_{20})=3$ and all the other states have zero cost. 
    For example, for state $s_{10}=\{ \{8\},  \{  (\{(2,8)\},3)    \} ,\{a\}    \}$, 
    we have $C_I(s_{10})= C(3)=5-3=2 $ as $ \textsc{Rev}(s_{10})=\{ (\{(2,8)\},3)   \}$. 
    For this example, 
    since there is only one secret state and the system is acyclic, for the sake of simplicity, 
    we omit   all successor states from  $s_{13}, s_{16}, s_{17}$ and $s_{20}$ as these states do no contribute to the cost/value function.

    Next, we iteratively update the value of each state in $T_{total}$ by value iterations.
    The value iteration procedure is shown in Table.\ref{tab:my_label}, which converges in $10$ steps. 
    Then the value of the initial $Y$-state  $V^*(s_1)=2$ is the optimal cost that we can achieve.
    Then we use Algorithm~2 to solve Problem 2.
    The resulting supervisor is shown in  Fig.~\ref{Solution3}.
    At the initial $Y$-state, we have $C_{rem}=V^*(s_1)=2$ and 
    $Dec_{V^*}^{max}(s_1,2)=Dec_{V^*}(s_1,2)=\{ \{ a \}  \}$. 
    Then we choose control decision $\{a\}$ and move to   $Z$-state $s_2$. 
    Note that  cost $C_{rem}$ remains unchanged as $C_I(s_{2})=0$. 
    Then observable event $o_1$ occurs and the system moves to $s_3$.
    At state $s_3$, we have $Dec_{V^*}^{max}(s_3,2 )=\{ \{ a \}  \}$. 
    By choosing control decision $\{a\}$,  we reach $s_4$ and we still have $C_{rem}=2- C_I(s_{4})=2$.
    Then  the system moves to $s_5$ upon the occurrence of event $o_2$.
    At state $s_5$, we have $Dec_{V^*}^{max}(s_5, 2)=\{ \emptyset \}$, then by choosing $\emptyset$ we reach $s_6$. 
    We repeat the above process, and obtain the supervisor shown in   Fig.~\ref{Solution3}, 
    which is an optimal and maximally-permissive solution. 
    \begin{table}
        \centering
        \begin{tabular}{c|c|c|c|c|c|c| c}
        \hline
         $V_i(q)$&0 &1 &2 &3 &\ldots &10& 11\\\hline
         $s_1$   &0 &0 &0 &0 &\ldots &2 & 2\\\hline
         $s_2$   &0 &0 &0 &0 &\ldots &2 & 2\\\hline
         $s_3$   &0 &0 &0 &0 &\ldots &2 & 2\\\hline
         $s_4$   &0 &0 &0 &0 &\ldots &2 & 2\\\hline
         $s_5$   &0 &0 &0 &0 &\ldots &2 & 2\\\hline
         $s_6$   &0 &0 &0 &0 &\ldots &2 & 2\\\hline
         $s_7$   &0 &0 &0 &0 &\ldots &2 & 2\\\hline
         $s_8$   &0 &0 &0 &2 &\ldots &2 & 2\\\hline
         $s_9$   &0 &0 &2 &2 &\ldots &2 & 2\\\hline
         $s_{10}$&0 &2 &2 &2 &\ldots &2 & 2\\\hline
         $s_{11}$&0 &0 &0 &5 &\ldots &5 & 5\\\hline
         $s_{12}$&0 &0 &5 &5 &\ldots &5 & 5\\\hline
         $s_{13}$&0 &5 &5 &5 &\ldots &5 & 5\\\hline
         $s_{14}$&0 &0 &0 &0 &\ldots &4 & 4\\\hline
         $s_{15}$&0 &0 &0 &4 &\ldots &4 & 4\\\hline
         $s_{16}$&0 &0 &4 &4 &\ldots &4 & 4\\\hline
         $s_{17}$&0 &4 &4 &4 &\ldots &4 & 4\\\hline
         $s_{18}$&0 &0 &0 &3 &\ldots &3 & 3\\\hline
         $s_{19}$&0 &0 &3 &3 &\ldots &3 & 3\\\hline
         $s_{20}$&0 &3 &3 &3 &\ldots &3 & 3\\\hline
        \end{tabular}
        \caption{Value iterations of $T_{total}$.}
        \label{tab:my_label}
    \end{table}
	\end{myexm}
	\section{Conclusion}\label{sec:7}
	In this paper, we systematically investigate both qualitative and quantitative synthesis of privacy-enforcing supervisors based on the notion of infinite-step opacity. 
	For the qualitative case, 
	we define a new class of bipartite transition systems that captures the delayed information in the control synthesis problem over a game structure.
	Based on the BTS, we proposed an effective algorithm  that  solves the standard infinite-step opacity control problem without the assumption that all controllable events are observable, which is restrictive and required by the existing work. 
	For the quantitative case, we propose the notion of secret-revelation-time as a quantitative measure for infinite-step opacity. 
	By suitably augmenting the timing information into the BTS, we solve the quantitative synthesis problem as an optimal total-cost control problem. 
	In the future, we plan to extend our results to the stochastic setting by considering the expectation of the secret-revelation cost. 	
	\bibliographystyle{plain}
	\bibliography{des}
		
\end{document}